\documentclass[journal,twoside]{IEEEtran}
\usepackage{amsmath,amsthm,amsfonts,amssymb,latexsym,extarrows,booktabs,array,arydshln}
\usepackage{fancyhdr,graphicx,tabularx}
\usepackage{multirow}
\usepackage{makecell}
\usepackage{graphicx}
\usepackage{amsfonts}
\usepackage{graphicx}
\usepackage{subfigure}
\usepackage {amsmath}
\usepackage {amssymb}
\usepackage {latexsym}
\usepackage {cite}
\usepackage {subfigure}
\usepackage {url}
\usepackage {stfloats}
\usepackage {mathrsfs}
\usepackage{algorithm}
\usepackage{algorithmic}
\usepackage{xcolor}
\usepackage{longtable}

\usepackage{setspace}
\IEEEoverridecommandlockouts
\newcommand{\ls}[1]
    {\dimen0=\fontdimen6\the\font
     \lineskip=#1\dimen0
     \advance\lineskip.5\fontdimen5\the\font
     \advance\lineskip-\dimen0
     \lineskiplimit=.9\lineskip
     \baselineskip=\lineskip
     \advance\baselineskip\dimen0
     \normallineskip\lineskip
     \normallineskiplimit\lineskiplimit
     \normalbaselineskip\baselineskip
     \ignorespaces
}

\newtheorem{proposition}{Proposition}

  \newtheorem{theorem}{Theorem}

\begin{document}
 \date{}
\title{Quantum  Learning  Based Nonrandom Superimposed Coding for Secure \\ Wireless  Access in 5G URLLC}
	\vspace{-10pt}
\author{Dongyang~Xu
	and Pinyi~Ren
	\vspace{-10pt}
 }
 \maketitle 
 \begin{abstract}
Secure wireless access in ultra-reliable low-latency communications (URLLC), which is a critical aspect of 5G security, has become increasingly  important due to its potential support of grant-free configuration. In  grant-free  URLLC, precise  allocation of  different pilot resources to different users that share the  same time-frequency resource is essential for the next generation NodeB (gNB) to exactly identify   those users  under access  collision and to maintain precise channel estimation required for reliable  data transmission. However, this process easily suffers from  attacks on pilots. We in this paper propose a quantum learning  based nonrandom  superimposed coding  method to encode and decode pilots on  multidimensional resources, such that the uncertainty of attacks  can be learned  quickly and eliminated precisely. Particularly, multiuser pilots for uplink access are encoded  as distinguishable subcarrier activation patterns (SAPs) and gNB decodes pilots of interest from  observed SAPs, a superposition of SAPs from access users, by  joint design of attack mode detection  and user activity detection  though a quantum learning network (QLN). We found that the uncertainty lies in the  identification process of codeword digits from the attacker, which can be always modelled as a black-box model, resolved by a quantum learning algorithm and quantum circuit.  Novel analytical closed-form expressions of failure probability are derived to characterize the reliability of this  URLLC system  with short packet transmission.  Simulations  how that our method can bring ultra-high reliability and low latency despite  attacks on pilots.
 \end{abstract}
\section{Introduction}
\label{introduction}
In 5G systems, ultra-reliable low-latency communications (URLLC)  has been a basis technology for an entirely new family of mission-critical use cases, such as  intelligent transportation, remote healthcare, industrial automation and so forth~\cite{TS22_261}. As the 3rd Generation Partnership Project (3GPP) claims, URLLC is required to meet the requirements of  reliability (99.9999$\%$ higher or more) and  latency in the order of 1ms or less. To remedy to this, 3GPP has been using a brute-force approach centered on system-level evaluation  and  a plethora of techniques including short packet transmission~\cite{Durisi}, grant-free mechanisms~\cite{R1_1703329}, leveraging spatial, frequency, and temporal diversity techniques~\cite{Bennis,Mahmood,TR38_824}. Among them,  grant-free   multiple  access (GFMA), which deletes the procedures of scheduling request and  scheduling grant,  has been deemed as a key enabler of URLLC.


In grant-free URLLC, procedures for uplink wireless access  have changed significantly  with respect to  traditional access procedures driven by pre-scheduling~\cite{K_S_Kim}.  Fig.~\ref{Attack_issue} depicts a general framework of grant-free access procedures  under radio resource control (RRC) connected state,   including  user activity detection (UAD),  channel estimation and data transmission. During uplink access, the next generation NodeB (gNB) needs to perform UAD before channel estimation to identify the activities of users that attempt to  access URLLC services~\cite{Liu_Y}. However, collision will happen if two or more users that share the same resources perform the uplink access at the same time, and in this case the gNB may not be able to successfully decode all of the uplink transmissions~\cite{TR38_881}.  In view of this, 3GPP specification has regulated the use of pilot-based user identification to distinguish colliding users and the use of pilot-based channel estimation to provide  channel information required for decoding data. According to the specification,  the gNB needs to allocate different distinguishable pilot resources to different users that share the same  resource when collision happens, which means that  the process of UAD, channel estimation and data transmission  all rely on  pilot sequences that  are distinguishable from each other.

\begin{figure}[!t]
	\centering \includegraphics[width=1.00\linewidth]{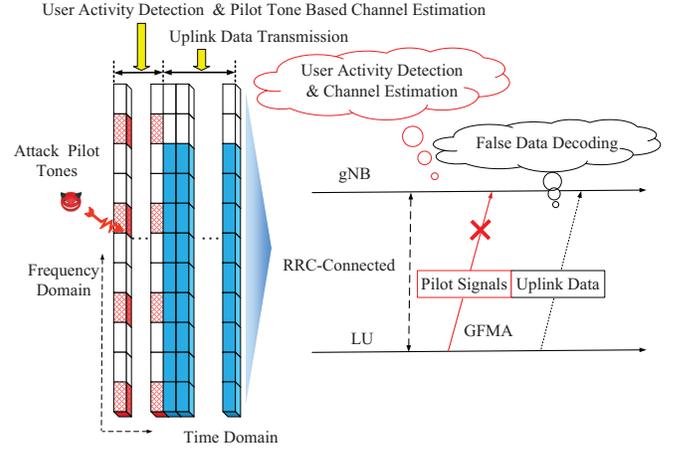}
	\caption{Illustration of uplink access procedures in grant-free URLLC and its security issue. }
	\vspace{-15pt}
	\label{Attack_issue}
\end{figure}

However, the recent research have shown that pilot-based mechanisms are currently suffering from a high risk of pilot-aware attack, a physical-layer threat that can acquire,  jam, spoof and null pilot sequences of interest~\cite{Shahriar}.   Indeed, paralyzing  uplink access through  tampering with pilots, preferred by attack,  is much  easier and more efficient than directly disturbing data transmission~\cite{TCC_Clancy}.  When a pilot aware attack occurs in grant-free URLLC,  UAD  and channel estimation become imprecise, indirectly causing large-scale low-reliability  huge-latency  data transmission in the uplink. Basically, the above risk stems from the vulnerability that  configuration parameters of pilots, including value/frequency location/bandwidth/subcarrier space of pilot sequences, usually  decoded from system information before RRC connected,  are   publicly-known without security  protection of upper layer mechanism, like authentication and key agreement,   throughout  current 5G new radio protocol and  easily compromised under pilot aware attack. 

Security for 5G URLLC   has  been studied early by 3GPP System Aspects Working Group 2 (SA2) in~\cite{TR33_825} and later included in R16~\cite{TR33_501}, stressing  the necessity of security mechanisms to deal with all potential security risks of URLLC on user-plane data transmission. As expected, the protection of  pilots in the uplink access has been a critical, complex and  unsolved aspect of security for grant-free URLLC. Note that URLLC is for user plane traffic according to TR 23.725~\cite{TR23_725} and any  upper-layer cryptographic  optimization during the very initial registration is not pursued for URLLC services~\cite{TR33_825}. This  imposes novel challenges on the uplink access subject to pilot-aware attack, rooted in the following aspects:
\begin{enumerate}
	\item \textit{Pilot Protection}: The  maintenance of distinguishability among deterministic pilot sequences is critical but very hard when they are jammed, spoofed and nulled successfully.
	\item \textit{Complex Access Design}: On one hand, the loss of distinguishability among  pilot sequences under attack makes both UAD and channel estimation invalid. On the other hand,  any scheme trying to protect  pilots has to necessarily but toughly  support the functionalities of  precise  UAD and reliable channel estimation. 
	\item \textit{Performance Evaluation in URLLC Scenarios}: 
	It is necessary to  re-evaluate the influence of uplink access design on  reliability and latency performance and to examine whether or not the design satisfy the basic requirements of URLLC services.
\end{enumerate}

\subsection{Related Works}
Recently, it has been pointed out in~\cite{Park} that  5G  URLLC requires  a tight coupling between ultra-high reliability and low latency to maintain the same level of security and be efficient in doing so at the same time. The former implies need for stringent security guarantee while the latter  implies extremely fast security provision. Progressing towards above goal  in uplink grant-free URLLC system  needs to tackle lots of  open  issues, including GFMA,  UAD, channel  estimation and so forth~\cite{Park}. These problems, as  specified both  in  standardization and academic communities, revolve around  the design of pilot, a  critical factor to support precise UAD and channel estimation in  grant-free URLLC~\cite{TR38_881,Jiang,Chen,Y_Li}.  In this situation, it is worthwhile to analyze and defend against the threat of  pilot aware attack in order to  meet the performance requirements for current 5G URLLC. Authors in~\cite{Q_Qi} have  investigated the  influence of pilot aware attack on GFMA  and  further stress the importance of pilot protection.

The countermeasures against pilot aware attack  have been rarely  researched in grant-free URLLC  while  studied a lot  in grant-based  systems with non-URLLC services  after the seminal work in~\cite{Clancy_denial,Sodagari_jamming}.  The insightful research in this area provide a natural basis for security design  in grant-free URLLC system, which can be divided into two main categories, one adopting  deterministic-pilot  based approach~\cite{Kapetanovic1,Tugnait2,Wu1} and the other one using random-pilot based approach~\cite{Sodagari_randomization,Xu_ICC}. The main point of the deterministic-pilot  based approach is to insist on storing and using publicly known and deterministic pilots  between transceivers. In this case,  the attack with strong ability of stealing prior information is quite hard to avoid and the removal of attack is difficult unless discarding the data. What's more, the countermeasures were limited to solely detecting the existence of pilot-aware attack, rather than further eliminating the attack. Random-pilot based approach was studied to  discover security advantages on random signal domain, which provides a potential pathway for recovering  true pilots on probabilistic domain. Its security performance was determined by the process of  transmission and retrieval of pilot (TRP). 
Basically, TRP requires that users transmit a set of  random pilot  sequences from a stored set and then gNB retrieves pilots from the received signals  by  testing whether or not a pilot sequence of interest is a right member of the stored set. 

Initially, the main focus of TRP is on  time domain and the design challenge is that random pilot signals transmitted over the air might be completely hidden by random channel taps  and attack signals, causing information confusion.   Despite the fact that adopting spatial domain property, e.g., natural separation of spatial correlation of massive-antenna
arrays, can reduce the effect of attack, the attacker can adjust its correlation property to be similar with one user of interest by roughly revising its spatial location~\cite{Adhikary}. At this point, the attacker would become more well-directed, rather than less effective.

To  address above challenges, many research turn to focus on time-frequency-code domain and  encode subcarrier activation pattern (SAP)  as robust medium to support TRP~\cite{Xu_ICC,ICA_SBDC,Zhenghao_Zhang}. Specifically, subcarriers  are activated and encoded as binary non-zero digits  when user transmits power on selected subcarriers. Otherwise when user keeps silent on them, those subcarriers are deactivated and  deemed as zero digits.  Therefore, various SAP candidates can be created and encoded by selectively activating and deactivating  subcarriers. In this way, each of  pilot signals can be carefully encoded as a unique binary codeword which corresponds to a unique SAP. When pilot signals from all users add up in the air and are received by gNB, the gNB needs to get aware of the subcarrier assignments and inspect the signal power on subcarriers to retrieve the code information. However, random pilot on time-frequency-code domain introduced a novel type of attack-defense domain on which the attack mode observed by gNB is transformed from a pilot-aware attack into a hybrid mode that embraces silence cheating and purely jamming attack~\cite{Xu_ICC}.  Nevertheless, the feasibility of applying  
SAP coding method  to secure TRP  against pilot-aware attack has been verified in ~\cite{Xu_ICC} in single-user grant-based scenario whereas  the scheme can hardly support multiple access.  

Actually, the concept of  SAP coding has been researched in many literatures in URLLC scenarios. Coding SAP for TRP  accords with the URLLC design principle, where exploiting  multi-domain diversity has been identified as a promising enabler of URLLC~\cite{Azari}. Furthermore, coding SAP  can create   a frequency-domain non-square packet structure which has been used as a baseline in 5G new radio since it can minimizes the transmission latency~\cite{Park}. 
Authors in~\cite{Hamamreh}  have concluded that encoding subcarrier activation patterns (SAPs) to carry  information in parallel with data transmission can improve the system reliability  with  less power and complexity, making it a very suitable candidate for 5G URLLC services. This technique was also studied in  grant-free URLLC systems in~\cite{Chen_URLLC}. However, the coding structure embedded in SAPs cannot directly be employed for pilot protection due to its lack of resilience against attack on code. The design of novel SAP coding method for pilot protection  has not attracted sufficient  attention  in grant-free URLLC.
\subsection{Motivations and Contributions}
The hints  above  motivate us to design coding principle of SAP 
for secure  uplink access in grant-free URLLC. To this end,  we should carefully consider the  access procedures under the circumstance of random pilots, random channels and a hybrid attack. SAP coding should be able  to protect pilots and at the same time support access procedures. This  refers to  three basic  aspects:
\begin{enumerate}
	\item Design novel access procedures  under SAP coding.
	\item Optimize coding principle to  resolve the effect of pilot aware attack.
	\item Evaluate the reliability and latency performance of novel grant-free URLLC system. 
\end{enumerate}
Considering the need for attack defense, we firstly introduce the functionality of attack mode detection (AMD) into uplink access, making access procedures embrace five parts, i.e., AMD, UAD, TRP, channel estimation and data transmission. Indeed, those procedures are expected to highly rely  on  pilots.  Secondly, we exploit  the framework of nonrandom superimposed coding on SAPs to protect pilots. 
In coding theory, the typical nonrandom superimposed code which was first considered by Kautz-Singleton was a set of binary vectors having the property that no vector is contained in a boolean sum (i.e. bitwise OR) of a small number of others~\cite{Kautz}. However, the original  nonrandom superimposed code  cannot provide resilient  security attributes. The key task before  us is to optimize the code on one hand and make it compatible with new access procedures on the other hand. To achieve the objective, we focus on the joint design of AMD, UAD and TRP through a five-layer quantum learning network (QLN), a precise and computing-efficient solution to the design. Our network not only allows  precise   AMD and UAD, but also supports  quick learning  and elimination of  decoding uncertainty  caused by  attack so that TRP can be protected well. The motivation behind QLN  is based on the following three key observations: 1) SAP coding brings to the nonrandom superimposed code signal  feature which is a useful resource for decoding enhancement but rarely exploited; 2)The  decoding uncertainty results from  the unauthentic codeword digits  generated by SAPs from the attacker. If those digits are found and erased, the  attack will be eliminated completely. 3) Quantum learning has the ability of finding the function of  uncertainty quickly. The specific goal of QLN is to develop a multi-layer decoding network that can exploit signal features and  quantum learning  to  remove the decoding uncertainty of nonrandom superimposed code under attack. To achieve this goal, the network is required to design AMD and  UAD  by signal features. Based on this design, the function of  identification of unauthentic codeword digits from the attacker can be modelled as a black-box model which is further resolved by quantum parallelism and interference.  Key contributions of our work are summarized  as follows:
\begin{enumerate}
	\item \textit{Access Design:} This work is the first to present a detailed framework of joint design of  AMD, UAD, TRP, channel estimation and data transmission under attack on pilots. 
	\item \textit{QLN Design:} We show how to use a quantum learning layer to model the  identification of unauthentic codeword digits from the attacker during TRP as a black-box boolean function, which is the  core of QLN design. Two distinguishable  binary quantum states  are  proved to exist as  the inputs of black-box boolean function and the measurement at the output could always get the global information from the two inputs. Thanks to this characteristic, the model can be learned precisely via a quantum circuit and the  decoding uncertainty caused by attack can be eliminated completely.  In our work, quantum learning is for the first time proved to have the potential of protecting nonrandom superimposed code   threatened by attack. Achieving  this potential requires the help of signal features.  This finding is important for future decoding method design since it provides a framework of modelling and resolving  attack uncertainty imposed on information coding. 
	\item \textit{Performance Analysis:} Novel  analytical closed-form expressions of failure probability   of grant-free URLLC system are derived to measure the reliability   in the regime of large-scale  antenna arrays and short data packets.  With these efforts, we can finally show how  quantum  learning  based nonrandom superimposed code  can efficiently protect the uplink access in grant-free URLLC.
\end{enumerate}
\emph{Organization:} In Section~\ref{Preliminaries}, we review the basic concepts  of nonrandom superimposed  coding and quantum learning. System model is given in~\ref{SyM}.  Design of QLN for secure uplink access  under attack  is  described  in details in  Section~\ref{QLNSDUA}.  We also present the design principle of quantum learning layer  in Section~\ref{DPQLL} and provide the performance analysis  in Section~\ref{PA}.  Numerical results are presented in Section~\ref{NR} and finally  we conclude our work in Section~\ref{Conclusions}.

 \section{Preliminaries}
 \label{Preliminaries}
 \subsection{Nonrandom Superimposed  Coding}
 In coding theory, the situation commonly encountered is that
 the codewords must be restored from partial information, like
 defected data (error correcting codes), or some superposition of
 the strings. These requirements lead to the technique of nonrandom superimposed coding.
 The typical nonrandom superimposed code which was first considered by
 Kautz-Singleton is a set of binary vectors having the property
 that no vector is contained in a boolean sum (i.e. bitwise OR)
 of a small number of others. Nonrandom superimposed coding has
 been widely used in various areas, such as, information retrieval
 for constructing signature files, multi-channel access for
 wireless communications and so forth.
 
 A nonrandom superimposed code of length $B$ and order $K$ has three properties:
 \begin{enumerate}
 	\item Each codeword  has a unique identifier.
 	\item Each  superposition of up to  $K$ different codewords is  unique and  each of the superimposed codeword can be  correctly decomposed   into a unique set of $K$ codewords.
 	\item Each of decomposed codewords from the superimposed codeword  is identified uniquely.
 \end{enumerate}
 With nonrandom superimposed coding, an information transmitter is enabled to express a set of information sequences as discrete  codewords.  Multiple codewords, when superimposed with each other, e.g., under wireless channel environment or in storage system,  can be decoded precisely  at  receiver according to the pre-shared codebook. Then the receiver tests whether or not a  codeword of interest is a right member of the stored set. In order to analyze the coding performance, we consider using the maximum distance separable code~\cite{Singleton} to  construct the nonrandom superimposed code. The specific  construction method  can be found in~\cite{Kautz}. The rate of  superimposed code  of  length  $B$ and cardinality $C$, denoted by $ R_{\rm c}$, is defined  by:
 \begin{equation}
 	R_{\rm c}{\rm{ = }} {{{{{\log }_2}C}} \mathord{\left/
 			{\vphantom {{\left( {{{\log }_2}C} \right)} {B}}} \right.
 			\kern-\nulldelimiterspace} {B}}
 \end{equation}
 where  $ B = q\left[ {1 +  K\left( {k - 1} \right)} \right],C = {q^k}, q \ge  K\left( {k - 1} \right)  \ge 3, K\ge 2$.
 
 \subsection{Quantum Learning}
 In quantum information processing, a quantum algorithm is an algorithm which runs on a quantum circuit model of computation~\cite{A_Montanaro}. In order to  describe the step-by-step computation,  it is necessary for any quantum algorithm to develop a suitable quantum oracle, also named as ``black-box" oracle. Particularly, if the unknown functions for  describing a specific problem is expected to be learned well by quantum algorithms, formulating a precise model of  quantum oracle from that  problem is the first and most-critical step. There exist lots of quantum algorithms, depending on the specific problems in that application area. One of the first examples of quantum algorithms is Deutsch-Jozsa algorithm which  is exponentially faster than any possible deterministic classical algorithm~\cite{D_Deutsch}. The algorithm requires 
 two qubits to distinguish  constant  functions in a black-box problem. A quantum algorithm with single-qubit input  is recently proposed to solve the same black-box problem in Deutsch algorithm~\cite{Z_Gedik}.
 
 Quantum learning  is a technique of  extracting  information from the quantum oracle of quantum algorithm to determine whether or not the oracle has some specific function properties~\cite{A_Cross}.  Typically,  such an quantum oracle  is  assumed to belong to some particular a priori fixed class $\cal C$ of possible functions. This model of quantum learning differs from other attempts to use quantum computers to perform machine learning tasks The reason behind the success of quantum learning in this area is the existence of hidden problem structure that quantum computers can exploit in ways that classical computers cannot.  Finding such hidden structure in other problems of practical interest remains an important open problem~\cite{A_Montanaro}.  For example, there is a class
 of functions that is polynomial time learnable from quantum
 coherent queries but not from classical queries~\cite{R_Servedio}.
 
 In this paper,  a concept $f$ over ${\left\{ {0,1} \right\}^n}$ is defined by a Boolean function $f:{\left\{ {0,1} \right\}^n} \to \left\{ {0,1} \right\}$ and a concept class $\cal C$ is a collection of concepts over ${\left\{ {0,1} \right\}^n}$.  
 Quantum learning  for a concept class $\cal C$ is defined as a sequence of unitary transformations ${U_1}$,${U_f}$,${U_2}$,${U_f}$, ${U_3}$,${U_f}$, $ \ldots $, where each ${U_i}$ is a fixed unitary transformation without any dependence on the concept.
 ${U_f}$ denotes  the quantum membership oracle which  is a transformation  acting on the computational basis states by mapping $\left| {x,y} \right\rangle  \mapsto \left| {x,y \oplus f\left( x \right)} \right\rangle $ where $x \in {\left\{ {0,1} \right\}^n}$ and $y \in \left\{ {0,1} \right\}$. 
 Quantum learning  provides a quantum state that encodes a hidden function, and its goal is to discover the function efficiently, meaning with a number of queries and an amount of postprocessing that scales polynomially in the number of input bits.
 
 \begin{figure}[!t]
 	\centering \includegraphics[width=1.00\linewidth]{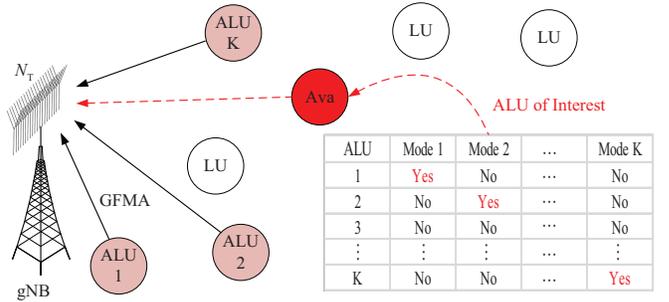}
 	\caption{Illustration of grant-free URLLC under piot randomization. }
 	\vspace{-15pt}
 	\label{System_model_sm}
 \end{figure}
 
\section{System Model}
\label{SyM}
\subsection{System Description of Grant-Free URLLC}
We consider an uplink single-input multiple-output (SIMO) orthogonal frequency
division multiplexed (OFDM)  system with a $N_{\rm T}$-antenna gNB (named as Alice)  and $G$ single-antenna legitimate users (LUs), indexed by the  set $\cal G$ with $\left| {\cal G} \right|=G$. As shown in Fig.~\ref{System_model_sm}, this system aims to provide grant-free URLLC services for $G$ uplink LUs. We assume that the buffer queues of all LUs are empty after initial system configuration.  When new bursty traffic arrives, a total of $K$ LUs, indexed by  set $\cal K$ with $\left| {\cal K} \right|=K$, will be activated immediately and then perform  GFMA to Alice  simultaneously.   This is done  by transmitting pilot tones and uplink short packet data. For simplicity, we denote those LUs needing  bursty traffic transmission as active LUs (ALUs). Correspondingly, Alice  is required to perform UAD, channel estimation and uplink data decoding. To  this end, traditionally, Alice is  supposed  to arrange deterministic pilot tones across subcarriers following Block type to support UAD and channel estimation~\cite{Ozdemir}. However, this set-up which is publicly known to all parties  inevitably  attracts a single-antenna attacker Ava, able to launch a pilot-aware attack (i.e., pilot jamming, spoofing, and nulling attack) on  pilot subcarriers to paralyze UAD and channel estimation, further causing data transmission outage. To defend against this  attack, Alice and ALUs  are required to adopt pilot randomization technique to randomize the values of pilots on each subcarrier. At Alice, three basic functionalities are configured, including  
\begin{enumerate}
	\item \textbf{AMD}:  Due to pilot randomization, a novel type of attack-defense domain is  introduced  on which the attack mode observed by Alice is transformed from a pilot-aware attack into a hybrid mode~\cite{Xu_ICC}. The  model of attack can be seen in subsection~\ref{AMA}. In such a scenario, AMD which aims to distinguish among different  attacks  is  a critical and necessary functionality. 
	\item \textbf{UAD}: Rather than relying on deterministic pilots, each ALU has to use random pilot as a temporal identifier to  represent itself, compelling Alice to redesign  UAD  such that the exact number and identities of ALUs can be identified.
	\item \textbf{TRP}: As previously talked, the traditional method of TRP cannot hold true under attack. In this paper, pilot signal from each  ALU is  encoded as a unique SAP. During uplink access, transmission of multiple pilots simultaneously makes SAPs overlap with each other. Alice needs to decode the observed SAPs to retrieve  pilots. The model of TRP can be seen in subsection~\ref{MTRP}. 
\end{enumerate}

It is required for Alice to jointly design  AMD,  UAD and TRP before channel estimation and data transmission. The  relationship among them can be seen Fig.~\ref{System_Framework_relation}. We assume that the joint design of AMD,  UAD and TRP occurs on time-frequency resource grid (TFRG)$\#$1 defined by a resource block containing $x$  subcarriers, $y$ OFDM  symbols and $z$ antennas with $\left( {x,y,z} \right)=\left( {N_{\rm E},m_{\rm E},N_{\rm T}} \right)$. The $N_{\rm E}$  subcarriers on TFRG $\#$1 are indexed by the set  ${\Psi }_{\rm E}$ with  $\left| {\Psi}_{\rm E}  \right|=N_{\rm E}$.  For the sake of easy calculation in the following subsections, we define ${\Psi _{\rm{E}}} \buildrel \Delta \over = \left\{ {1, \ldots ,{N_{\rm{E}}}} \right\}$. Channel estimation occurs on TFRG$\#$2  defined by a resource block satisfying  $\left( {x,y,z} \right)=\left({{N}_{\rm CE},m_{\rm E},N_{\rm T}} \right)$ where the ${N}_{\rm CE}$  subcarriers on TFRG $\#$2 are indexed by the set  ${\Psi }_{\rm CE }$ with  $\left| {\Psi}_{\rm CE} \right|={N}_{\rm CE}$.  Data transmission works on  TFRG$\#$3 defined by a resource block satisfying  $\left( {x,y,z} \right)=\left({{N}_{\rm D},m_{\rm D},N_{\rm T}} \right)$ where the $N_{\rm D}$  subcarriers on TFRG$\#$3 are indexed by the set ${\Psi }_{\rm D}$ with $\left| {\Psi}_{\rm D}  \right|=N_{\rm D}$.   There exist ${\Psi}_{\rm E} \cap {\Psi}_{\rm CE}  = \emptyset $ and ${\Psi}_{\rm E} \cap {\Psi}_{\rm D}  = \emptyset$.   $\Delta f$  and  $T_{\rm s}$ respectively denote the frequency-domain subcarrier spacing and  OFDM symbol time, both  following the choices  in 5G  NR numerologies~\cite{Park}.  The total latency $T$ during uplink access satisfies the following condition:
\begin{equation}
	T_{\rm con}\ge T = \left( {m_{\rm E}+ m_{\rm D}} \right)\times {T_{\rm s}} + {T_{\rm  extra}}
\end{equation}
where ${T_{\rm  extra}}$ denotes the time occupied by operations other than channel estimation and data transmission.  The maximum value $T_{\rm con}$ of total latency  is limited to  $T_{\rm con}=1\times 10^{-3}s$.

\begin{figure}[!t]
	\centering \includegraphics[width=1.00\linewidth]{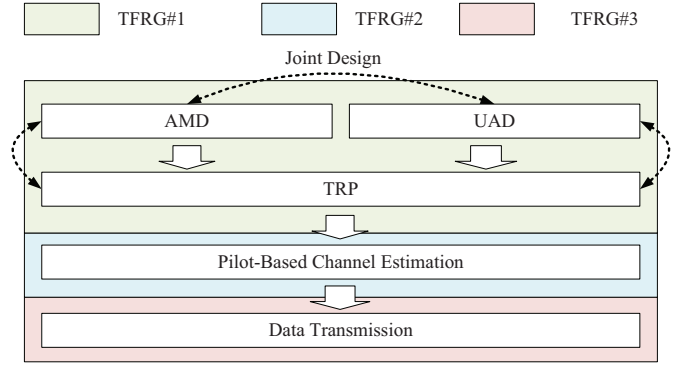}
	\caption{Relationship among AMD, UAD, TRP, channel estimation and data transmission. }
	\vspace{-15pt}
	\label{System_Framework_relation}
\end{figure}

\subsection{Random Pilot Signal Model}
We denote $x^{i}_{{\rm{L,}}m}\left[ k \right] $ and  $x^{i}_{{\rm{A}}}\left[ k \right] $ respectively as the pilot values for the $m$-th ALU and Ava at the $i$-th  subcarrier and $k$-th symbol time.  $x^{i}_{{\rm{L,}}m}\left[ k \right] $ satisfies  $x^{i}_{{\rm{L,}}m}\left[ k \right] = {x_{{\rm{L,}}m}}\left[ k \right] = \sqrt {{\rho _{{\rm{L,}}m}}} {e^{j{\phi _{k,m}}}},\forall i, i \in {\Psi_{\rm{E}} \cup {\Psi}_{\rm CE} }, {m \in {{\cal G}}}, {\phi}_{k, m} \in {\cal A} $ where ${\cal A}$ is a set  
satisfying $\left\{ {\phi :\phi  = {{2m\pi } \mathord{\left/
			{\vphantom {{2m\pi } C}} \right. 
			\kern-\nulldelimiterspace} C},0 \le m \le C - 1,C = \left| \cal A \right|} \right\}$.  SAP coding for TRP means how to encode  ${\phi}_{k, m}$ as SAPs able to be decoded by Alice, which can be seen in the  subsection~\ref{MTRP}. Note that ${x_{{\rm{L,}}m}}\left[ k \right] $ can be superimposed onto a dedicated pilot sequence having been optimized under a non-security oriented scenario. The new pilot sequence is then utilized for channel estimation. At this point, ${\phi}_{k, m}$ is an additional phase difference for security consideration.  We do not impose any prior constraints on $x^{i}_{{\rm{A}}}\left[ k \right] $ which satisfies  $x^{i}_{{\rm{A}}}\left[ k \right] = \sqrt {{\rho _{\rm{A}}}} {e^{j{\varphi _{k,i}}}},i \in {\Psi_{\rm{E}}}$. ${{\rho _{{\rm{L}},m}}}$ and ${{\rho _{\rm{A}}}}$ respectively denote the transmitting power of pilot signals from the $m$-th ALU and Ava. 

\begin{figure}[!t]
	\centering \includegraphics[width=1\linewidth]{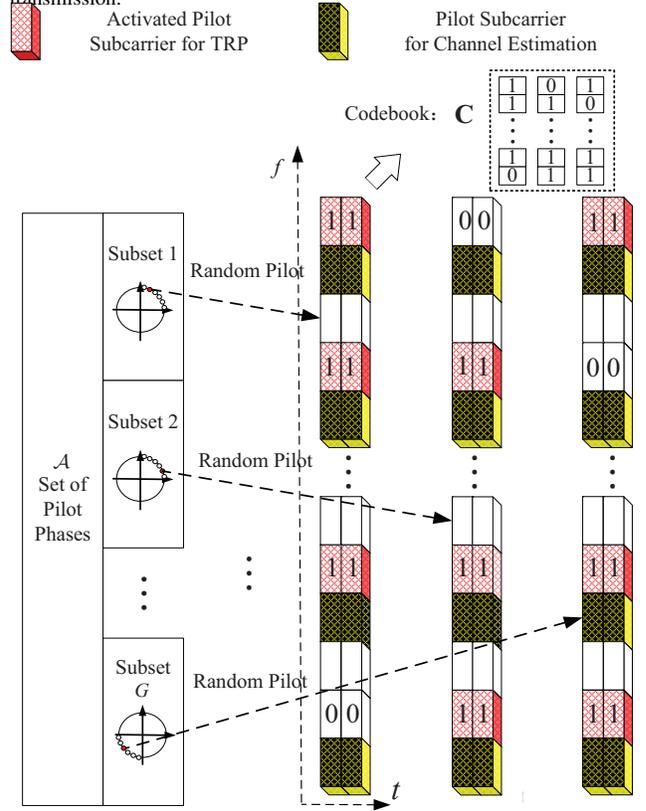}
	\caption{Illustration of pilot transmission through  SAP coding.}
	\vspace{-10pt}
	\label{Bloom_expressing}
\end{figure} 

\begin{figure}[!t]
	\centering \includegraphics[width=1\linewidth]{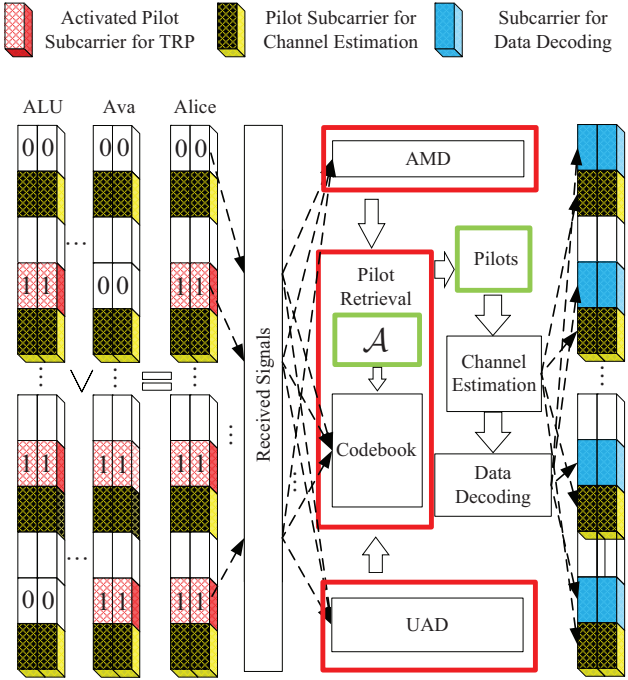}
	\caption{Illustration of  framework-level description of AMD, UAD, TRP, channel estimation and data decoding.}
	\vspace{-15pt}
	\label{Bloom_decoding11}
\end{figure} 
\subsection{Model of TRP}
\label{MTRP}
The nonrandom superimposed code of length $N_{\rm E}$, size $C$ and order $K+1$ for encoding SAPs  is denoted by  a $N_{\rm E}\times C$ binary matrix  ${\bf{B}}{\rm{ = }}{\left[ {{{\bf b}^{j}}} \right]_{ {\rm{1}} \le j \le C}}$, where each codeword ${\bf{b}}^{j}{\rm{ = }}{\left[ {{b_{i,j}}} \right]_{{\rm{1}} \le i \le N_{\rm E}}}$ is a column vector.   On TFRG$\#$1, nonrandom superimposed code  based TRP requires that each of codewords in the codebook  be  mapped as a unique element in set $\cal A$.  Besides this, codebook partition and allocation  must be done prior to TRP. The codeword vectors in ${\bf{B}}$  are spilt into $G$ independent codeword clusters in which the $i$-th  cluster contains $\left[ {{C \mathord{\left/
			{\vphantom {C G}} \right.
			\kern-\nulldelimiterspace} G}} \right]$  codeword vectors indexed by  the set  ${\cal B}_{i}$. Correspondingly,   $\cal A$ is equally divided into $G$ subsets and we denote the $i$-th  subset by ${\cal A}_i$. The $i$-th codeword cluster  constitutes a sub-matrix denoted  by   $\left[ {{{\bf{b}}^{j \in {{\cal B}_{i}}}}} \right]$ which together with ${\cal A}_i$ are  exclusively allocated  to the $i$-th LU. 

When the $i$-th ALU in the set $\cal K$ attempts to use random pilot phase ${\phi}_{k, i}$  in ${\cal A}_i$ for TRP at the $k$-th OFDM symbol time, it will  select the corresponding  codeword ${{\bf{b}}_i}$  in the $i$-th sub-matrix  with  ${{\bf{b}}_i}\in \left[ {{{\bf{b}}^{j{'} \in {{\cal B}_i}}}} \right], i \in \cal K$ and further  generate its  intended SAP according to the codeword ${{\bf{b}}_i}, i\in \cal K$. 
The principle of generating SAPs via  ${{\bf{b}}_i}, i \in \cal K$ is that if the $j$-th digit of ${{\bf{b}}_i}, i \in \cal K$ is equal to 1, the pilot signal is inserted on the $j$-th subcarrier  and otherwise this subcarrier will be idle. 
According to this principle, $K$ ALUs would  independently express their pilot phases in the form of  SAPs encoded by $K$ codewords.  This process can be shown in Fig.~\ref{Bloom_expressing}.

Multiple SAPs from different ALUs, after undergoing wireless channels, would suffer from the superposition interference from each other. The codeword superposition is modelled  by:
\begin{equation}
	{{\bf{b}}_1} \vee  \cdots  \vee {{\bf{b}}_K}={{\bf{b}}_{{\rm{S}},K}}, {{\bf{b}}_{{\rm{S}},K}} \vee {\bf{a}}={{\bf{b}}_{\rm{I}}}
\end{equation}
where $ \vee $ denotes the boolean sum operation and  ${{\bf{b}}_i},1\le i\le K$ denotes the   codeword exploited by the $i$-th  ALU. ${\bf{a}}={\left[ {{a_i}} \right]_{1 \le i \le B}}$ denotes the codeword from Ava. The specific elements of ${\bf{a}}$ can be seen in Eq.~(4). ${{\bf{b}}_{{\rm{S}},K}}$ denotes the superposition version of codewords from ALUs and can be hardly acquired directly by Alice under attack.  ${{\bf{b}}_{\rm{I}}}$ denotes  the  codeword that corresponds to the SAP  observed by Alice. Note that  ${{\bf{b}}_{\rm{I}}}$ cannot be derived directly by Alice because what Alice can observe are solely signals on subcarriers.  Alice needs to retrieve the information of ${\phi}_{k, i}, i \in {\cal K}$ based  on the SAP.  As we can see,  the above mathematical model represents a complete process of TRP. The  framework-level description of AMD, UAD, TRP, channel estimation and data decoding is  shown in Fig.~\ref{Bloom_decoding11}.

\subsection{Attack Modelling and Analysis}
\label{AMA}
To destroy the process of pilot retrieval on multiple subcarriers,  Ava would launch a hybrid attack mode including:
\begin{enumerate}
	\item \textbf{Silence Cheating (SC):} Ava  keeps silence to misguide Alice since Alice cannot  recognize the non-existence of attacks.
	\item \textbf{Wide-Band Pilot (WB-PJ):} Ava  transmits its random signals on  the whole available subcarriers to  cover  the  codeword information of all ALUs. 
	\item \textbf{Partial-Band Pilot Jamming (PB-PJ):} Ava transmits its random signals on arbitrary  part of the whole subcarriers  such that the codeword decoded by Alice is imprecise.
\end{enumerate}
Under above model, the  elements of vector ${\bf{a}}$,  determined by  attack modes, can be given by:
\begin{equation}
	{\bf{a}} = \left\{ {\begin{array}{*{20}{c}}
			{{{\left[ {\begin{array}{*{20}{c}}
								0& \cdots &0
						\end{array}} \right]}^{\rm{T}}}}&{ \rm SC}\\
			{{{\left[ {\begin{array}{*{20}{c}}
								1& \cdots &1
						\end{array}} \right]}^{\rm{T}}}}&{\rm WB\text{-}PJ}\\
			{{{\left[ {\begin{array}{*{20}{c}}
								0& \cdots &1
						\end{array}} \right]}^{\rm{T}}}}&{\rm PB\text{-}PJ}
	\end{array}} \right.
\end{equation}
Note that Ava does not know the codeword used by ALU of interest since codewords from ALUs are randomly selected. Furthermore,  the distribution of above attack modes is  unknown to Alice and LUs. Therefore,  the coding principle of SAP at Ava  is to use random coding on SAPs.  By examining Eq. (3) and Eq. (4), we can  easily test  the ability of nonrandom superimposed code  in resisting attack. For example, let us consider three codewords from  ALUs, i.e., $ \left[ {\begin{array}{*{20}{c}}
		1&0&0 &1
\end{array}} \right]$, $ \left[ {\begin{array}{*{20}{c}}
		0&0&1 &1
\end{array}} \right]$, $ \left[ {\begin{array}{*{20}{c}}
		1&0&1 &0
\end{array}} \right]$ and  a  codeword from Ava, i.e., $\left[ {\begin{array}{*{20}{c}}
		1&1&0 &0
\end{array}} \right]$.  Since signal energy from Ava  is dispersed on pilot subcarriers, what  Alice could observe is a WB-PJ attack, which indicates no any useful information and imposes huge confusion  on SAP decoding. 

\subsection{Pilot Based Channel Estimation}
Consider the basic OFDM procedure on TFRG$\#$2. Pilot tone vectors  of ALUs and Ava over $N_{\rm E}$ subcarriers are respectively stacked as $N_{\rm E}$ by 1 vector ${{\bf{x}}_{{\rm{L,}}m}}\left[ k \right] = \left[ {{x^{j}_{{\rm{L,}}m}}\left[ k \right]} \right]^{\rm T}_{j \in {\Psi}_{\rm CE} }$ and  ${{\bf{x}}_{{\rm{A}}}}\left[ k \right] = \left[ {{x^{j}_{{\rm{A}}}}\left[ k \right]} \right]^{\rm T}_{j \in {\Psi}_{\rm CE} }$. We consider using  orthogonal pilots, i.e., $ {{\bf{x}}_{{\rm{L}},m}}{\bf{x}}_{{\rm{L}},n}^ + = 0,\forall m \ne n$.
The length of cyclic prefix is assumed to be larger than the maximum number $L$ of channel taps. The parallel streams, i.e., ${{\bf{x}}_{{\rm{L,}}m}}\left[ k \right] $ and ${{\bf{x}}_{{\rm{A}}}}\left[ k \right]$ are modulated with inverse fast Fourier transform. Then the time-domain  $N_{\rm E}$ by 1 vector  ${{\bf{y}}^i}\left[ k \right]$, derived by Alice after removing the cyclic prefix at the $i$-th receiving antenna, can be written as:
\begin{equation}\label{E.1}
	{{\bf{y}}^i}\left[ k \right] = \sum\limits_{m \in {\cal K}}^{} {{\bf{H}}_{{\rm{C,}}m}^i{{\bf{F}}^{\rm{H}}}{{\bf{x}}_{{\rm{L,}}m}}\left[ k \right]}  + {\bf{H}}_{{\rm{C,A}}}^i{{\bf{F}}^{\rm{H}}}{{\bf{x}}_{\rm{A}}}\left[ k \right] + {{\bf{v}}^i}\left[ k \right]
\end{equation}
Here,  ${\bf{H}}_{{\rm{C}},m}^i$  is the $N_{\rm CE} \times N_{\rm CE}$ circulant matrices of the $m$-th ALU, with the  first column  given by ${\left[ {\begin{array}{*{20}{c}}
			{{\bf{h}}_{{\rm{L,}}m}^{{i^{\rm{T}}}}}&{{{\bf{0}}_{1 \times \left( {N_{\rm CE}- L} \right)}}}
	\end{array}} \right]^{\rm{T}}}$.
${\bf{H}}_{{\rm{C,A}}}^i$ is a $N_{\rm CE} \times N_{\rm CE}$ circulant matrix with the first column given by ${\left[ {\begin{array}{*{20}{c}}
			{{\bf{h}}_{\rm{A}}^{{i^{\rm{T}}}}}&{{{\bf{0}}_{1 \times \left( {N_{\rm CE} - L} \right)}}}
	\end{array}} \right]^{\rm{T}}}$ and ${\bf{h}}_{\rm{A}}^i $ is  assumed to be independent with  ${\bf{h}}_{{\rm L},m}^i, \forall m \in \cal{K}$. ${\bf{h}}_{{\rm L},m}^i$ and ${\bf{h}}_{{\rm A}}^i$  denote the $L \times 1$ channel impulse response vectors, respectively from the $m$-th ALU and Ava to the $i$-th receiving antenna of Alice.  $L$ denotes the length of channel taps.  ${\bf{F}} \in {{\mathbb C}^{N_{\rm CE} \times N_{\rm CE}}}$ denotes the discrete Fourier transform matrix. ${{\bf{v}}^i}\left[ k \right]\sim {\cal CN}\left( {0,{\sigma ^2}{{\bf{I}}_{N_{\rm CE}}}} \right)$ denotes the noise vector on time domain at the $i$-th antenna of Alice within the $k$-th symbol time. ${\sigma ^2}$ is the average noise power of Alice. Taking fast Fourier transform, Alice finally derives the frequency-domain $N_{\rm E}$ by 1 signal vector at the $i$-th receiving antenna as:
\begin{equation}\label{Eq.3_3}
	{{\widehat{\bf{y}}}^i}\left[ k \right]= \sum\limits_{j \in {\cal K}}^{} {{{\bf{F}}_{\rm{L}}}{\bf{h}}_{{\rm{L}},j}^i{x_{{\rm{L}},j}}\left[ k \right]}  +
	{\bf I}^{i}\left[ k \right]
\end{equation}
where  ${{\bf{F}}_{\rm{L}}} = \sqrt {{N_{\rm{CE}}}} {\bf{F}}_{{\rm{*}}}$. ${\bf{F}}_{{\rm{*}}}$ denotes the first $L$ columns of  ${\bf{F}}$.  $ {\bf I}^{i}\left[ k \right]$ satisfies: 
\begin{equation}
	{\bf I}^{i}\left[ k \right]={\rm{Diag}}\left\{ {{{\bf{x}}_{\rm{A}}}\left[ k \right]} \right\}{{\bf{F}}_{\rm{L}}}{\bf{h}}_{\rm{A}}^i + {{\bf{w}}^i}\left[ k \right]
\end{equation}
where ${{\bf{w}}^i}\left[ k \right] = {\bf{F}}{{\bf{v}}^i}\left[ k \right]$.
Stacking ${\widehat {\bf{y}}^i}\left[ k \right]$ within $K$ OFDM symbol time, we can rewrite  signal model in Eq.~(\ref{Eq.3_3}) as:
\begin{equation}
	{{\bf{Y}}}_{{\rm{}}}^i = \sum\limits_{j \in {\cal K}}^{} {{{\bf{F}}_{\rm{L}}}{\bf{h}}_{{\rm{L}},j}^i{{\bf{x}}_{{\rm{L}},j}}}  +
	{\bf I}^{i}
\end{equation}
where the $N_{\rm CE} \times K$   matrix ${\bf{Y}}_{{}}^i$ satisfies ${\bf{Y}}_{{\rm{}}}^i = \left[ {\widehat {\bf{y}}^i{{\left[ k \right]}_{ k \in \cal K}}} \right]$ and ${\bf I}^{i} $ satisfies ${\bf{I}}_{{\rm{}}}^i = \left[ { {\bf{I}}^i{{\left[ k \right]}_{ k \in \cal K}}} \right]$.
The $1 \times K$ vector ${{\bf{x}}_{{\rm{L}},m}}$ satisfies ${{\bf{x}}_{{\rm{L}},m}}{\rm{ = }}\left[ {{{x}_{{\rm{L}},m}}{{\left[ k \right]}_{ k \in \cal K}}} \right]$. 

Provided that TRP can be protected well, the random pilot of interest, i.e., $m$-th one,  can be known  by Alice without being affected by Ava. A least square  estimation of ${{{\bf{h}}^{i}_{{\rm L},m}}}$, denoted by ${{{\widehat{\bf{h}}}^{i}_{{\rm L},m}}} $ can be given by:
\begin{equation}\label{E.4_4}
	{{{\widehat{\bf{h}}}^{i}_{{\rm L},m}}} = \left\{ {\begin{array}{*{20}{c}}
			{{\bf{h}}_{{\rm{L}},1}^i + {{\left(  {{\bf{F}}_{\rm{L}}} \right)}^ + } {{\bf{I}}^i}{{\left( {{{\bf{x}}_{{\rm{L}},1}}} \right)}^ + }}&{if\, m=1}\\
			{{\bf{h}}_{{\rm{L}},2}^i  +  {{\left(  {{\bf{F}}_{\rm{L}}} \right)}^ + }{{\bf{I}}^i}{{\left( {{{\bf{x}}_{{\rm{L}},2}}} \right)}^ + }}&{if\, m=2}\\
			\vdots & \vdots \\
			{{\bf{h}}_{{\rm{L}},K}^i +  {{\left(  {{\bf{F}}_{\rm{L}}} \right)}^ + }{{\bf{I}}^i}{{\left( {{{\bf{x}}_{{\rm{L}},K}}} \right)}^ + }}&{if\, m=K}
	\end{array}} \right.
\end{equation}
where $\left( {\cdot} \right)^+ $ is  the Moore-Penrose pseudoinverse.   

\begin{figure*}
	\begin{equation}\setcounter{equation}{13}
		\label{eq_17}
		\frac{1}{{{N_{\text{T}}}}}\widehat {\mathbf{g}}_{j,{m}}^{\text{H}}{{\mathbf{y}}_{{\rm d},j}}\left[ k \right] = \frac{{{d_{{\text{L,}}{m}}}\left[ k \right]}}{{{N_{\text{T}}}}}\sum\limits_{i = 1}^{{N_{\text{T}}}} {\widehat g_{j,{m},i}^*{g_{j,{m},i}}}  + \frac{{{d_{{\text{L,}}{m}}}\left[ k \right]}}{{{N_{\text{T}}}}}\sum\limits_{{p\in {\cal K}},p \ne {m}}^{} {\sum\limits_{i = 1}^{{N_{\text{T}}}} {\widehat g_{j,{m},i}^*{g_{j,p,i}}} }  + \frac{1}{{{N_{\text{T}}}}}\sum\limits_{i = 1}^{{N_{\text{T}}}} {\widehat g_{j,{m},i}^*{w_{j,i}}\left[ k \right]} 
		\vspace{-10pt}
	\end{equation}
\end{figure*}
\subsection{Data Transmission}
Without loss of generality, we assume: 1) $K$  ALUs   share the same  subcarriers at the same OFDM symbols within TFRG$\#$3 for short packet transmission;  2) $K$  ALUs operate at the same rate $R$ which can be calculated as  $R = \frac{B}{{{m_{\text{D}}}{T_{\text{s}}}{N_{\text{D}}}\Delta f}}$; 3) Alice employs  matched filter over $N_{\rm T}$ antennas  on each subcarrier.
Under above assumptions, the  receiving  signal model of Alice at the $j$-th subcarrier within the $k$-th OFDM symbol is denoted by ${{\bf{y}}_{{\rm d},j}}\left[ k \right]$, satisfying
\begin{equation}\label{Eq.6_6}\setcounter{equation}{10}
	{{\bf{y}}_{{\rm d},j}}\left[ k \right] = \sum\limits_{m \in {\cal K}}^{} {{{\bf{g}}_{j,m}}{d_{{\rm{L}},m}}\left[ k \right]}  + {{\bf{w}}_j}\left[ k \right]
\end{equation}
where ${{d_{{\rm{L}},m}}\left[ k \right]}, m \in {\cal K}$ denotes the data symbol transmitted by the $m$-th ALU at the $k$-th OFDM symbol  and satisfies ${\mathbb E}\left[ {{{\left| {{d_{{\rm{L}},m}}\left[ k \right]} \right|}^2}} \right] = \gamma $ where $\gamma $ is the instantaneous signal-to-noise-ratio (SNR). ${{\bf{w}}_j}\left[ k \right] = {\left[ {{w_{j,i}}\left[ k \right]} \right]_{1 \le i \le {N_{\rm{T}}}}}, j \in \Psi_{\rm D}$  denotes the noise vector at the $j$-th subcarrier within  $k$-th OFDM symbol and satisfies 
${{\bf{w}}_j}\left[ k \right]\sim {\cal CN}\left( {0,{{\bf{I}}_{{N_{\rm{T}}}}}} \right)$.  ${{\bf{g}}_{j,m}} = {\left[ {{g_{j,m,i}}} \right]_{1 \le i \le {N_{\rm{T}}}}}, j \in \Psi_{\rm D}, m \in {\cal K}$ denotes the $j$-th subcarrier channel vector  stacked  by the $m$-th ALU across $N_{\rm T}$ antennas, and satisfies: 
\begin{equation}
	{{\mathbf{g}}_{j,m}} = {\left[ {\begin{array}{*{20}{c}}
				{{{\mathbf{F}}_{{\text{L}},j}}{{\mathbf{h}}}_{{\text{L}},m}^1}& \cdots &{{{\mathbf{F}}_{{\text{L}},j}}{{\mathbf{h}}}_{{\text{L}},m}^{{N_{\text{T}}}}} 
		\end{array}} \right]^{\text{T}}}
\end{equation} 
where ${{\mathbf{F}}_{{\text{L}},j}}$ denotes the $j$-th row of ${{\mathbf{F}}_{{\text{L}}}}$.  We define ${{{\widehat{\bf{h}}}^{i}_{{\rm L},m}}}$ as the estimated version of ${{{{\bf{h}}}^{i}_{{\rm L},m}}}$ and can  derive the estimated version of ${{{\mathbf{g}}}_{j,m}}$  as  ${{\widehat{\mathbf{g}}}_{j,m}}$,  satisfying:
\begin{equation}
	{{\widehat{\mathbf{g}}}_{j,m}} =  {\left[ {{{\widehat g}_{j,m,i}}} \right]_{1 \le i \le {N_{\rm{T}}}}}= {\left[ {\begin{array}{*{20}{c}}
				{{{\mathbf{F}}_{{\text{L}},j}}{\widehat{\mathbf{h}}}_{{\text{L}},m}^1}& \cdots &{{{\mathbf{F}}_{{\text{L}},j}}{\widehat{\mathbf{h}}}_{{\text{L}},m}^{{N_{\text{T}}}}} 
		\end{array}} \right]^{\text{T}}}
\end{equation} 
Based on ${{\widehat{\mathbf{g}}}_{j,m}}$, Alice derives the matched filter as  $\frac{1}{{{N_{\text{T}}}}}\widehat {\mathbf{g}}_{j,m}^{\text{H}}$ which is then applied on Eq.~(\ref{Eq.6_6}) to decode  the data symbol of $m$-th ALU.  Receiving signal  weighted by $\frac{1}{{{N_{\text{T}}}}}\widehat {\mathbf{g}}_{j,m}^{\text{H}}$ can be expressed in Eq.~(\ref{eq_17}).  We assume  ${\widehat {\mathbf{g}}_{j,{m}}} = {{\mathbf{g}}_{j,{m}}}$ when  no estimation errors exist, and otherwise we assume ${\widehat {\mathbf{g}}_{j,{m}}} = \left( {1 - \lambda } \right){{\mathbf{g}}_{j,{m}}} - \lambda {\widetilde {\mathbf{g}}_{j,{m}}},0 < \lambda  < 1$ where  ${\widetilde {\mathbf{g}}_{j,{m}}}\sim {\cal CN}\left( {0,{{\bf{I}}_{N_{\rm T}}}} \right)$ is independent with ${\widehat {\mathbf{g}}_{j,{m}}}$ and larger $\lambda $ means that the estimation gets worse.

\section{QLN for Secure Uplink Access Under Attack}
\label{QLNSDUA}
In this section, we concentrate  on the joint design of  AMD, UAD and TRP through  a five-layer QLN. The QLN includes a initial layer, input layer, hidden layer, quantum learning  layer and output layer.  Details on the design of each layer are given in Fig.~\ref{Quantum_learning} and discussed in the following subsections.

The design principle of layers in QLN depends on the processing procedures for decoding SAP. Since what Alice can directly observe are solely signals on subcarriers that follow an unexpected SAP, decoding the SAP needs to achieve a signal-feature based transformation from signals to binary code information, a complex and multi-step process following a fixed order. Hence it is natural to adopt multiple layers, each layer responsible for one specific functionality. QLN uses different layers to abstract increasingly complex procedures, i.e., an initial layer is to model signal information, an input layer to extract signal feature, an hidden layer to deal with features etc. This often leads to better generalization ability towards future access design in 5G and Beyond. Besides, this set-up can help gNB easily learn more detailed functionalities within each layer and more abstract relationships among different layers. It is theoretically possible to represent any possible functionalities with a single layer QLN. In order to achieve single layer QLN, it is required to change the method of SAP coding correspondingly, which could be a future research direction. Universal approximation which states that the least layer required to approximate any SAP decoding  procedures can be also researched in the future.  In this paper, we will provide an easy-to-implement and initial version of QLN, for the sake of making clear the basic design principle in this area. 

\begin{figure}[!t]
	\centering \includegraphics[width=1.0\linewidth]{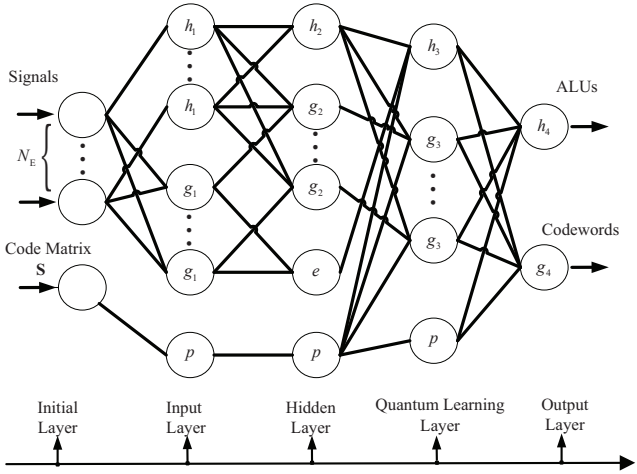}
	\caption{Illustration of a five-layer QLN.}
	\vspace{-15pt}
	\label{Quantum_learning}
\end{figure} 
\subsection{Initial Layer }
The functionality of initial layer is to provide data information required by the network and to give a mathematical model describing those data well. This layer is constituted by ${N_{\rm E}}+1$ initial nodes among which ${N_{\rm E}}$ initial nodes are responsible for collecting and modelling wireless signals on resource block  and another node stores  all the available codebook matrices related to the nonrandom superimposed code. 
\subsubsection{Data  on the First ${N_{\rm E}}$ Initial Nodes}
The encoding principle  of SAPs determines  that  ${N_{\rm E}}$ subcarriers are expected to be randomly occupied by random pilot signals. When $K$ ALUs attempt to access, we define the possible number of signals on the $i$-th subcarrier as $c_{i}$. Then the signals within the $k$-th OFDM symbol on the $i$-th initial node of  this layer can be denoted by ${ {\mathbf{y}} _i}\left[ k \right]$, satisfying:
\begin{equation}\setcounter{equation}{14}
	{ {\mathbf{y}} _i}\left[ k \right] = \left\{ {\begin{array}{*{20}{c}}
			{{{{\mathbf{w}} }_i}\left[ k \right]} \\ 
			{{\mathbf{g}} _{{\text{E}},i}^{}{n_i}\left[ k \right] + {{{\mathbf{w}} }_i}\left[ k \right]} \\ 
			{\sum\limits_{j = 1}^{c_{i}} {{\mathbf{g}} _{{\text{}}j,m}^{}{x_{{\text{L}},j}}\left[ k \right] + } {{ {\mathbf{w}} }_i}\left[ k \right]} \\ 
			{\sum\limits_{j = 1}^{c_{i}} { {\mathbf{g}} _{{\text{}}j,m}^{}{x_{{\text{L}},j}}\left[ k \right] +  {\mathbf{g}} _{{\text{E}},i}^{}{n_i}\left[ k \right] + } {{ {\mathbf{w}} }_i}\left[ k \right]} 
	\end{array}} \right.
\end{equation}
where 
${{\mathbf{g}} _{{\text{E}},i}^{}}$ denotes  the $N_{\rm E} \times 1$  channel vector from  Ava to Alice at the $i$-th pilot subcarrier  across $N_{\rm T}$ antennas, and satisfies:
\begin{equation}
	{{\mathbf{g}}_{{\rm E},j}} = {\left[ {\begin{array}{*{20}{c}}
				{{{\mathbf{F}}_{{\text{L}},j}}{{\mathbf{h}}}_{{\text{A}}}^1}& \cdots &{{{\mathbf{F}}_{{\text{L}},j}}{{\mathbf{h}}}_{{\text{A}}}^{{N_{\text{T}}}}} 
		\end{array}} \right]^{\text{T}}}
\end{equation}
${{n_i}\left[ k \right]}$ denotes the interfering  signal from Ava at the $i$-th pilot subcarrier within the $k$-th OFDM symbol. 

Let us explain the Eq.~(14) in details.  ${ {\mathbf{y}} _i}\left[ k \right]={{\bf{w}}_i}\left[ k \right] $ holds true when no pilot signals arrive; If interfering signals from Eva are received by Alice, there will be ${{\mathbf{y}} _i}\left[ k \right] = {\mathbf{g}} _{{\text{E}},i}^{}{n_i}\left[ k \right] + {{\mathbf{w}} _i}\left[ k \right]$; Otherwise if  pilot signals from ALUs arrive at Alice, ${{\mathbf{y}} _i}\left[ k \right] = \sum\limits_{j = 1}^{{}} { {\mathbf{g}} _{{\text{L}},j,i}^{}{x_{{\text{L}},j}}\left[ k \right] + } {{\mathbf{w}} _i}\left[ k \right]$ holds true; Finally, if  signals from both ALUs and Eva coexist, there will be ${ {\mathbf{y}} _i}\left[ k \right] = \sum\limits_{j = 1}^{} { {\mathbf{g}} _{{\text{L}},j,i}^{}{x_{{\text{L}},j}}\left[ k \right] +  {\mathbf{g}} _{{\text{E}},i}^{}{n_i}\left[ k \right] + } { {\mathbf{w}} _i}\left[ k \right]$. 

By stacking signals across $M$ OFDM symbols at the $i$-th pilot subcarrier, Alice derives the data captured by  the  $i$-th initial node  as a $N_{\rm T} \times M$ matrix ${{\bf{Y}}_i}$, given by:
\begin{equation}
	{{\bf{Y}}_i}{\rm{ = }}\left[ {\begin{array}{*{20}{c}}
			{{{\bf{y}}_1}\left[ k \right]}&{{{\bf{y}}_2}\left[ k \right]}& \cdots &{{{\bf{y}}_M}\left[ k \right]}
	\end{array}} \right]
\end{equation}
where the parameter $M, M >  m_{\rm E}$ is unknown now and will be given in the following section, depending on specific  design requirements. 
\subsubsection{Data on the $\left( {{N_{\rm{E}}} + 1} \right)$-th Initial Node}
We  define the ${B \times \left( {\begin{array}{*{20}{c}}
			C\\
			k
	\end{array}} \right)}$  matrix  ${{\bf{B}}_k}, k = 2, 3,..., K+1$  as the collection of all of the boolean sums of  codewords from $\bf B$, taken exactly $k$ at a time.   Each column vector of ${{\bf{B}}_k}$   represents a unique  codeword.
The  $\left( {{N_{\rm{E}}} + 1} \right)$-th node is responsible for storing the stacked matrix of superimposed code, denoted by ${\bf{S}}$, satisfying:
\begin{equation}
	{\bf{S}} = \left[ {\begin{array}{*{20}{c}}
			{\bf{B}}&{{{\bf{B}}_2}}& \cdots &{{{\bf{B}}_{K + 1}}}
	\end{array}} \right]
\end{equation}

As we can see, the final outputs of initial layer are ${{\bf{Y}}_i}, i \in {\Psi}_{\rm E}$ and $\bf S$.

\subsection{Input Layer}
The input layer  is responsible for signal feature extraction which provides  quantized signal  features necessary to SAP decoding. Three types of nodes  for achieving this functionality are prepared, including energy feature extraction (EFE) nodes, independence feature extraction (IFE) nodes and code storage (CS) node. 
\subsubsection{EFE Nodes} There exist a total  of $N_{\rm E}$ EFE nodes,  each node performing signal energy detection on the corresponding subcarrier so as to accurately determine the number of signals  on the subcarrier.

In order to perform signal energy detection  on the $i$-th pilot subcarrier, Alice needs to make the $i$-th EFE node  generate a normalized covariance matrix ${\widehat {{\bf{R}}}_i}$ that satisfies ${\widehat {{\bf{R}}}_i} = \frac{1}{{{\sigma ^2}}}{\bf{Y}}_{ i}{{{\bf{Y}}_{i}}^{\rm{H}}}$. The ordered eigenvalues of ${\widehat {{\bf{R}}}_i}$ are given  by  ${\lambda _{{M}}} >  \ldots  > {\lambda _1} > 0$. After that, Alice  constructs $M$ test statistics  as $T_{k} = \frac{{{\lambda _k}}}{{{\lambda _{1}}}}\mathop {\gtrless}\limits_{{{{\overline{\cal H}}_k}}}^{{{{{ {\cal H}}_k}}}} {\gamma}_k, 2\le k\le M$ where ${\gamma}_k$ denotes the  decision threshold of the $k$-th test statistic.  The hypothesis ${{ {\cal H}}_k}$  means $\left| {M+1- k} \right|$ signals coexist and  ${{\overline{\cal H}}_{k}}$ means the opposite. If Alice hopes to recognize $\left| {M+1- k} \right|, 3\le k \le M$  signals  precisely, both ${{ {\cal H}}_k}$  and ${{\overline{\cal H}}_{k-1}}$ should hold true. Otherwise if  Alice hopes to recognize $\left| {M-1} \right|$  signals, ${{ {\cal H}}_2}$ should hold true. Based on $M$  eigenvalues,  Alice cannot recognize $M$ signals.  
In fact, Alice needs to detect at most $K+1$ signals including $K$ signals from ALUs and one signal from Ava on each subcarrier. This requires $M= K+2$ such that the eigenvalue space of ${\widehat {{\bf{R}}}_i}$  is enough to capture  features of all possible $K+1$ signals.

For $T_{k}$,  the larger ${\lambda}_k$ is, the larger threshold  ${\gamma}_k$ is required to make precise decision. To detect all possible $K+1$ signals without false alarm, ${\gamma}_{K+2}$ is  utilized as the common  threshold for all other detectors, i.e.,  $T_{k}, 1 \le i\le K+1$.  
Given ${\widehat {{\bf{R}}}_i}$,  Eq. (49) in~\cite{Kobeissi}  provides  a   threshold  function $f\left( {{N_{\rm{T}}},K,{P_f}} \right)$  to  measure what level  the probability of false alarm, i. e.,  ${P_f}$, can achieve.    We should note that this function is not only a monotone decreasing  function of  ${P_f}$ but also a monotone increasing  function of $K$. For a given probability constraint ${\varepsilon}$, a lower bound $\gamma \left( {{\varepsilon}} \right)$ could be always expected with  $\gamma \left( {{\varepsilon}} \right) = f\left( {{N_{\rm{T}}},K,{\varepsilon}} \right)$. Under this equation,  zero ${\varepsilon}$ can be achieved by flexibly configuring  $\gamma \left( {{\varepsilon}} \right)$ and  $\gamma \left( {{0}} \right)$ can be determined exactly~\cite{Kobeissi}.  This phenomenon can be experimentally confirmed in Fig.~\ref{Detection} where  $N_{\rm T}$ is  fixed to be $128$ and $K$  is respectively configured as 12, 16, and 20.  Basically, this phenomenon originates from the well-known Marcenko-Pastur Law~\cite{Hoydis}. In this paper, we let ${\gamma}_{K+2} \buildrel \Delta \over = \gamma \left( {{0}} \right)$.

Based on above method for establishing the decision threshold, Alice constructs $K+2$ test statistics  as $T_{k} = \frac{{{\lambda _k}}}{{{\lambda _{1}}}}\mathop {\gtrless}\limits_{{{{\overline{\cal H}}_k}}}^{{{{{ {\cal H}}_k}}}} \gamma \left( {{0}} \right), 2\le k\le K+2$. After testing above statistics, Alice can derive the number of signals at the $i$-th subcarrier as $M_{i} $.  Finally, we rewrite the function of this EFE node as a function $h_1$ with the output ${M_i}$, satisfying:
\begin{equation}
	{M_i} = h_1\left( {\widehat {{\bf{R}}}_i}, {f\left( {{N_{\rm{T}}},K,0} \right)} \right)
\end{equation}

\subsubsection{IFE Nodes} There are a total of $N_{\rm E}$ IFE nodes and each IFE node performs inner-product  operation across all subcarriers so as to accurately capture  the  independence features embedded in pilot signals from ALUs and Ava.  

Our treatment of  independence characteristic of receiving signals  will refer to two aspects, that is,  wireless channels and random pilot signals. Particularly,  it is well known that wireless channels from different ALUs  are independent with each other due to the inherent constraints of antenna spacing. Thanks to  the  transparency transmission set-up in the uplink in current 3GPP specification, pilot signals  from different ALUs  are also independent with each other. 

In order to quantize above  independence features having embedded on the pilot subcarriers,  IFE nodes need to perform the mutual inner-product  operation among pilot signals at different frequency-domain positions. 	The inner-product  operation between  pilot signals at the $i$-th pilot subcarrier and those at the $j$-th pilot subcarrier can be  derived as:
\begin{equation}
	{d_{i,j}} = g_1\left( {{I_{i,j}}} \right), {I_{i,j}} = \left\langle {\frac{{{{ {\mathbf{y}} }_i}\left[ k \right]}}{{\left\| {{{ {\mathbf{y}} }_i}\left[ k \right]} \right\|}},\frac{{{{ {\mathbf{y}} }_j}\left[ k \right]}}{{\left\| {{{ {\mathbf{y}} }_j}\left[ k \right]} \right\|}}} \right\rangle, \forall k 
\end{equation}
where $ {d_{i,j}}$ denotes the differential code digit and $\left\langle {\cdot} \right\rangle $ denotes  inner product operation.  
$g_1\left( x \right) $ represents  the  encoder satisfying $g_1\left( x \right) = \left\{ {\begin{array}{*{20}{c}}
		0&{x \le \zeta  }\\
		1&{x > \zeta  }
\end{array}} \right.$ where   $\zeta $ denotes the  decision  threshold.  The value of ${ {\mathbf{y}} _i}\left[ k \right]$ at the $i$-th pilot subcarrier  can be seen in Eq.~(14).   To identify the principle of designing  $\zeta $
, we give the following interpretation. According to law of large numbers, the inner product between signals from two independent individuals approaches zero. On the contrary, the inner product between signals from the same node can reach a value with its amplitude equal to one. In theory, the value of  $\zeta $
can thus be configured to be a certain value, i.e., 0.5.

\begin{algorithm}[!t]
	\caption{Codeword Generation}
	\begin{algorithmic}[1]
		\label{Alg1}
		\FOR {$i=1$ to $i=N_{\rm E}$}
		\STATE The   inner product is performed  between  receiving  signals at  the  $i$-th pilot subcarrier  and receiving  signals across all ${N_{\rm E}}$  pilot subcarriers.   A code vector of length ${N_{\rm E}}$ is formulated.
		\STATE  Perform XOR operation between each digit of the code vector  and the  digit (usually 1 if signals exist, otherwise 0.) at the  $i$-th pilot subcarrier. The result is   ${{\bf{d}}_i}$.
		\ENDFOR
	\end{algorithmic}
\end{algorithm}
On this basis,  the independence features having embedded on  the $i$-th pilot subcarrier can be extracted as the binary code vector ${{\bf{d}}_i}$. See more details about the extraction process in Algorithm~\ref{Alg1}. Based on above algorithm, we define the function of the $i$-th IFE node as $ {g_1}$,  satisfying:
\begin{equation}
	\left[ {\begin{array}{*{20}{c}}
			i&{{{\bf{d}}^{\rm T}_i}}
	\end{array}} \right] = {g_1}\left( {{{\bf{Y}}_i}} \right), i = 1, \cdots ,{N_{\rm{E}}}
\end{equation} 
where ${{{\bf{Y}}_i}}$ and $\left[ {\begin{array}{*{20}{c}}
		i&{{{\bf{d}}^{\rm T}_i}}
\end{array}} \right]$ respectively denote the input and output of the $i$-th IFE node. 
\subsubsection{CS Node}
CS node is responsible for storing the data from the $\left( {{N_{\rm{E}}} + 1} \right)$-th initial node. We define the function of this node as function $p$, satisfying 
\begin{equation}
	{\bf{S}} = p\left( {\bf{S}} \right)
\end{equation}
\begin{figure}[!t]
	\centering \includegraphics[width=1\linewidth]{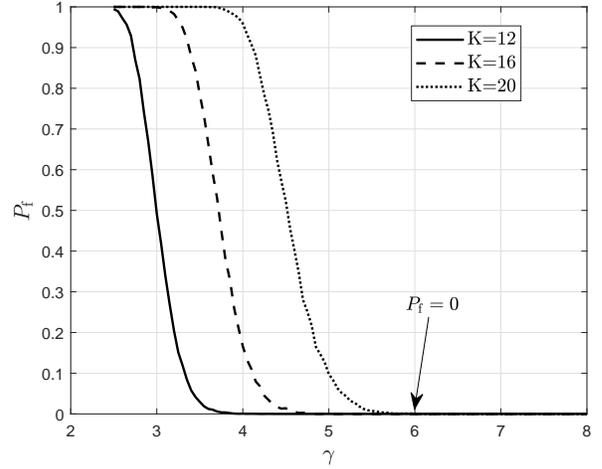}
	\caption{ ${P_f}$ versus  $\gamma$ under various  $K$ on an arbitrary  single subcarrier. }	
	\label{Detection}
	\vspace{-15pt}
\end{figure}
\subsection{Hidden Layer}
This layer is responsible for  dealing with those  quantized features having been derived in the previous layer. This  is done through three types of nodes, i.e., energy feature processing (EFP) node, independence feature processing (IFP) nodes and CS node. 

The function of  EFP node, denoted by $h_2$, is to encode the  detected number of signals  on $N_{\rm E}$ subcarriers  and output   a multivariate codeword vector $[\begin{array}{*{20}{c}}
	{{\bf{m}}_{\rm{I}}^T}&{{\bf{b}}_{\rm{I}}^T}
\end{array}]$, given by:
\begin{equation}
	[\begin{array}{*{20}{c}}
		{{\bf{m}}_{\rm{I}}^T}&{{\bf{b}}_{\rm{I}}^T}
	\end{array}] = {h_2}\left( {{M_1}, \cdots ,{M_{{N_{\rm{E}}}}}} \right)
\end{equation}
where 
\begin{equation}
	{{\bf{m}}_{\rm{I}}}{\rm{ = }}{\left[ {\begin{array}{*{20}{c}}
				{{M_1}}& \cdots &{{M_{{N_{\rm{E}}}}}}
		\end{array}} \right]^T}
\end{equation}	
and
\begin{equation}
	{{\bf{b}}_{\rm{I}}}{\rm{ = }}{\left[ {\begin{array}{*{20}{c}}
				{{1_{\left\{ {x\left| {x \ge 1} \right.} \right\}}}\left( {x = {M_1}} \right)}& \cdots &{{1_{\left\{ {x\left| {x \ge 1} \right.} \right\}}}\left( {x = {M_{{N_{\rm{E}}}}}} \right)}
		\end{array}} \right]^T}
\end{equation}		
where ${1_{\left\{ {x \ge 1} \right\}}}\left( x \right)$ is the indicator function. 

The IFP nodes have two functionalities. For the first functionality, ${N_{\rm{E}}}$ IFP nodes generate  ${{\bf{b}}_{\rm{I}}}$ according to the Eq. (24)  and then make the $j$-digit of  ${{ {\bf{b}}}_{{\rm{I}}}}$  zero to derive a  codeword denoted by ${\bf{a}} \left( {{j}} \right)$. The overall  functionality is denoted by  a function $g_2$, satisfying
\begin{equation}
	{\bf{a}}\left( j \right) = {g_2}\left( {{M_1}, \cdots ,{M_{{N_{\rm{E}}}}},j} \right),j = 1, \cdots ,{N_{\rm{E}}}
\end{equation}

The second functionality is to  stack  ${{\bf{d}}_j}, 1\le j \le N_{\rm E}$ as a matrix ${\bf{D}}$ satisfying  ${\bf{D}}{\rm{ = }}{\left[ {\begin{array}{*{20}{c}}
			{{{\bf{d}}_{{1}}}}& \cdots &{{{\bf{d}}_{{{{N_{\rm{E}}}}}}}}
	\end{array}} \right]}$. This task  is accomplished by  another IFP node and can be described by a function $e$, satisfying:
\begin{equation}
	{\bf{D}}{\rm{ = }}e\left( {{{\bf{d}}_{{1}}}, \cdots ,{{\bf{d}}_{{{{N_{\rm{E}}}}}}}} \right)
\end{equation}
The  functionality of CS node in this layer  is the same as that in the input layer.

\subsection{Quantum Learning Layer}
As we mentioned previously,  the core of joint design of AMD, UAD and TRP lies in how to find   the unauthentic codeword digits  generated by SAPs from the attacker. This is done by the quantum learning layer. To this end, this layer firstly configures a detection node to realize AMD  and then configures ${N_{\rm{E}}}$ identification nodes to  learn  the attacker's codeword digits. CS node is also configured to provide the matrix $\bf S$. 
\subsubsection{AMD} The functionality of AMD at detection node is built on the following three principles: (1) WB-PJ attack happens only when  $\bf D$ is a all-ones  matrix, (2)  PB-PJ attack happens when ${{\bf{b}}_{\rm{I}}} \notin {{\bf{B}}_K}$ or ${{\bf{b}}_{\rm{I}}} \in {{\bf{B}}_K}$, $\sum\limits_{i = 1}^K {\sum\limits_{j = 1}^{N_{\rm E}} {{b_{j,i}}} }\ne \sum\limits_{i = 1}^{N_{\rm E}} {{M_i}}$, (3)  SC happens when ${{\bf{b}}_{\rm{I}}} \in {{\bf{B}}_K}$ and $\sum\limits_{i = 1}^K {\sum\limits_{j = 1}^{N_{\rm E}} {{b_{j,i}}} }= \sum\limits_{i = 1}^{N_{\rm E}} {{M_i}}$. Following above principles, the function of detection node, denoted by $h_3$, satisfies:
\begin{equation}
	\left( {A,{{\bf{b}}_{\rm{A}}}} \right) = {h_3}\left( {{\bf{D}},{{\bf{b}}_{\rm{I}}},{{\bf{m}}_{\rm{I}}},{\bf{S}}} \right), A{\rm{ = }}\left\{ {\begin{array}{*{20}{c}}
			1&{{\rm{WB - PJ}}}\\
			0&{{\rm{SC}}}\\
			{{\rm{ - }}1}&{{\rm{PB - PJ}}}
	\end{array}} \right.
\end{equation}
where ${{\bf{b}}_{\rm{A}}}={{\bf{b}}_{{\rm{S}},K}}$ when  $A=0,1$ and otherwise ${{\bf{b}}_{\rm{A}}}={{{\bf{b}}_{\rm{I}}}}$.
\subsubsection{Identification of Codeword Digits from Ava} Let us define $B_j$ to indicate whether the  pilot signal on the $j$-th subcarrier is from Ava or not. When  $B_j$ is  equal to 1, Alice can determine that the  pilot signal on the $j$-th subcarrier is transmitted by Ava. Based on this principle, we can know that $B_j=1, \forall j, 1\le j \le N_{\rm E}$ holds true  when Alice identifies WB-PJ attack successfully since the codeword $\bf a$ of Ava in this case is a  vector of all ones. Provided that SC attack is identified successfully,  Alice can know that $\bf a$  would be a zero vector and $B_j=0, \forall j, 1\le j \le N_{\rm E}$ holds true. When PB-PJ attack is identified,  the situation is rather complex, discussed as follows:

If  ${{\bf{a}}} \in {\bf{B}}$, there  exist two possibilities including  ${{{\bf{b}}}_{{\rm{I}}}}\in {\bf{B}}_{K}$ and  ${{{\bf{b}}}_{{\rm{I}}}}\in {\bf{B}}_{K+1}$. For these two cases, Alice can precisely know which ALU is currently being interfered  by Ava. This is done by checking ${\bf{a}}\left( {{j}} \right),1 \le j \le N_{\rm E}$ to examine whether  ${\bf{a}}\left( {{j}} \right),1 \le j \le N_{\rm E}$ is a vector in ${\bf{B}}_{K}$ or not. If  the $j$-th one is,  the ${\bf d}_j$ is deemed to be the codewords of Ava and configures $B_j=1$. 

If $\bf a \notin {{\bf{B}}}$, there exists ${{{\bf{b}}}_{{\rm{I}}}}\notin {\bf{B}}_{K+1}$. The  discussion is divided into two aspects: 1) $\bf a $ can be included  by ${{\bf{b}}_{{\rm{S}},K}}$.  In this case, there will be ${{\bf{b}}_{\rm{I}}}={{\bf{b}}_{{\rm{S}},K}}$ and Alice can perfectly identify codewords from ALUs even though $\bf a$ is unknown. For simplicity, we define $\bf a$ as  a zero vector; 2) $\bf a $ cannot be not included  by ${{\bf{b}}_{{\rm{S}},K}}$. In this case,  there must exist ${{\bf{b}}_{\rm{I}}}\ne {{\bf{b}}_{{\rm{S}},K}}$ and a set ${\cal D}_{1}$ satisfying $M_{j\in {\cal D}_{1}}=1$. Alice needs to calculate ${\bf{a}}\left( {{j}} \right), j \in {\cal D}_{1}$ and examine whether ${\bf{a}}\left( {{j}} \right),1 \le j \le N_{\rm E}$ is a vector in ${\bf{B}}_{K}$ or not. If it is, the value of $j$ is stored in the set ${\cal D}_{2}$ with  ${\cal D}_{2}\subseteq {\cal D}_{1}$,  and Alice configures $B_{j \in {\cal D}_{2}}=1$.  In this way, Alice can always recognize  $\bf a $ as ${\bf d}_j, j \in {\cal D}_{2}$.
\begin{figure}[!t]
	\centering \includegraphics[width=1.00\linewidth]{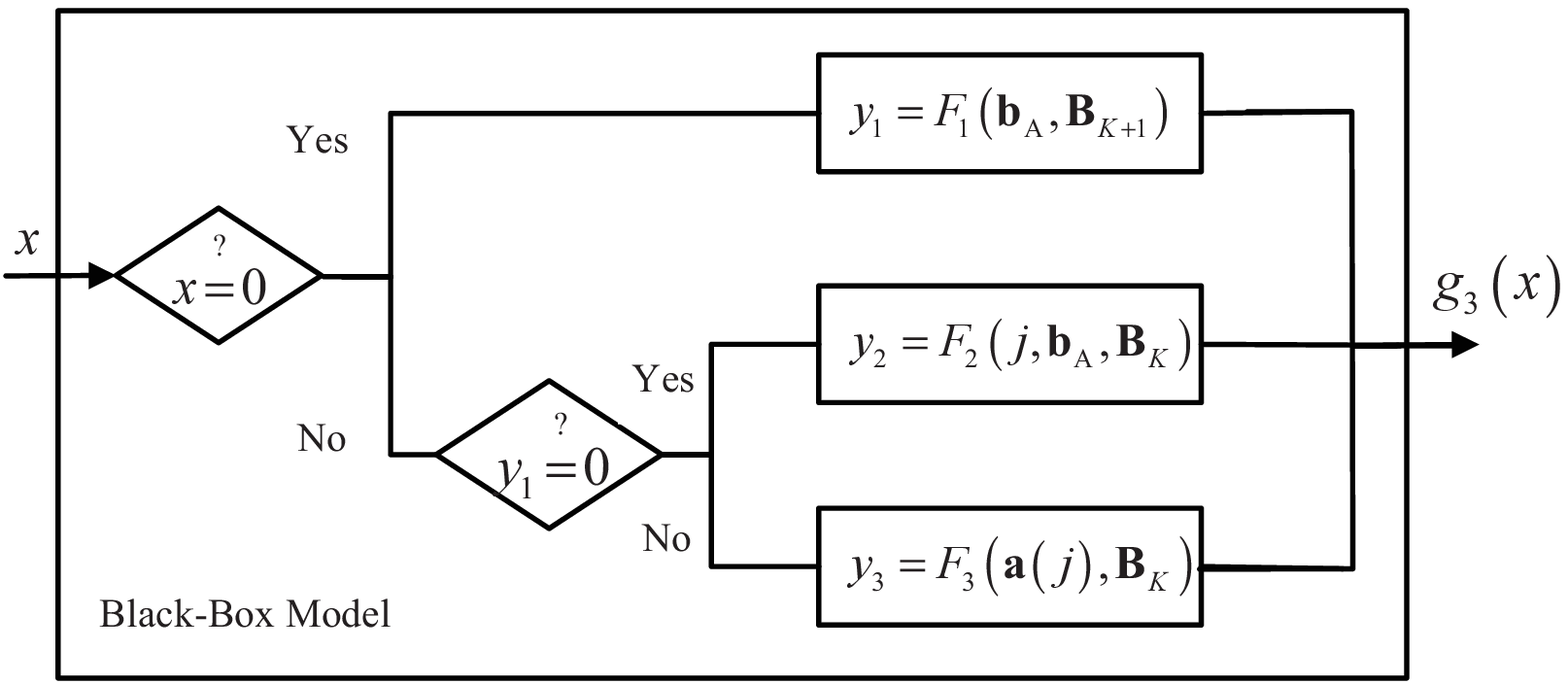}
	\caption{Diagram of the construction of black-box model with its input $x$ and output $g_3\left( x \right)$ satisfying $g_3:\left\{ {0,1} \right\} \to \left\{ {0,1} \right\}$.}
	\label{Black_Box}
	\vspace{-10pt}
\end{figure}

Based on above decoding principle,  we model the function of the $j$-th identification node on this layer as a function ${g_3}$ satisfying:
\begin{equation}
	B_j = {g_3}\left( {x,{{\bf{b}}_{\rm{A}}},{j},{\bf{a}}\left( {{j}} \right),{\bf{S}}} \right),j = 1, \ldots ,{N_{\rm{E}}}
\end{equation}
where the choices of values of $x$ can be 0 and 1, depending on the specific design. More specifically, when configuring $x=0$, Alice examines whether  there exists  ${{{\bf{b}}}_{{\rm{A}}}}\in {\bf{B}}_{K+1}$ by checking ${{{\bf{b}}}_{{\rm{A}}}}$ in codebook ${\bf{B}}_{K+1}$.  This process  is defined as  decision function ${y_1} = {F_1}\left( {{{\bf{b}}_{\rm{A}}},{{\bf{B}}_{K + 1}}} \right)$ which outputs 0  if  ${{\bf{b}}_{\rm{A}}} \in {{\bf{B}}_{K + 1}}$ and otherwise 1.
Otherwise when $x=1$,  Alice has two options, respectively represented by two subfunctions, i.e., ${y_2} = {F_2}\left( {{j},{{\bf{b}}_{\rm{A}}},{{\bf{B}}_{K}}} \right)$ if ${{{\bf{b}}}_{{\rm{A}}}}\in {\bf{B}}_{K+1}$ and ${y_3} = {F_3}\left( {{\bf{a}} \left( {{j}} \right),{\bf{B}}_{K}} \right)$ if ${{{\bf{b}}}_{{\rm{A}}}}\notin {\bf{B}}_{K+1}$.  The details can be shown as follows:
\begin{enumerate}
	\item${y_2} = {F_2}\left( {{j},{{\bf{b}}_{\rm{A}}},{{\bf{B}}_{K}}} \right)$ is a decision function. Firstly, it configures the ${j}$-th digit of ${{\bf{b}}_{\rm{A}}}$  to be zero and derives the revised codeword. Then it outputs 0 if  this revised codeword belongs to ${{\bf{B}}_K}$ and otherwise 1. 
	\item  ${y_3} = {F_3}\left( {{\bf{a}} \left( {{j}} \right),{\bf{B}}} \right)$ is a decision function which outputs 1 if ${\bf{a}} \left( {{j}} \right) \in {\bf{B}}$ and  otherwise 0. 
\end{enumerate}
As we can see, the above function describes the decoding principle well.  For the sake of convenience, in what follows,  we denote  ${g_3}\left( {x,{{\bf{b}}_{\rm{A}}},{j},{\bf{a}}\left( {{j}} \right),{\bf{S}}} \right)$ by  ${g_3\left( x \right)}$.

\subsection{Output Layer}
Based on the output of previous layer, this layer aims to perform UAD and to eliminate  the effect of codeword digits from Ava such that those codewords from ALUs can be recovered. Apart from the CS node, this layer contains two new nodes, one defined by the UAD node and the other one defined by decoding node.  The function of UAD node and decoding node  can be respectively modelled as ${h_4}$ and ${g_4}$, respectively  satisfying:
\begin{equation}
	{\cal K} = {h_4}\left( {{{\bf{b}}_{\rm{A}}},{B_1}, \ldots ,{B_{{N_{\rm{E}}}}},{\bf{S}}} \right)
\end{equation} 
and 
\begin{equation}
	\left( {\left\{ {{{\bf{b}}_i}\left| {i \in \cal K} \right.} \right\}} \right) =  {g_4}\left( {{{\bf{b}}_{\rm{A}}},{B_1}, \ldots ,{B_{{N_{\rm{E}}}}},{\bf{S}}} \right)
\end{equation} 
The details of above two functions are given as follows.
When WB-PJ attack has been identified in quantum learning layer, i.e.,  $B_j=1, \forall 1 \le i \le N_{\rm E}$, $M_i$ is updated as $M_i-1$, $1\le i \le N_{\rm E}$. Correspondingly, Alice encodes the number of signals on each subcarrier and derives   the novel codeword ${{\overline {\bf{b}}}_{{\rm{I}}}}$ which satisfies  ${{\overline {\bf{b}}}_{{\rm{I}}}}\in  {\bf{B}}_{K}$.  ${\bf{b}}_i,  i \in \cal K $ could be derived by decomposing  ${{\overline {\bf{b}}}_{{\rm{I}}}}$ based on the codebook matrix $\bf B$. 
If SC attack has been identified, i.e.,  $B_j=0, \forall 1 \le i \le N_{\rm E}$,  ${\bf{b}}_i,  i \in \cal K$ could be  derived  directly since ${{ {\bf{b}}}_{{\rm{I}}}}={{\bf{b}}_{{\rm{S}}, K}} \in  {\bf{B}}_{K}$.  The above two cases mean that the number and identities of  ALUs, together with their codewords,  can  be precisely determined under WB-PJ and SC attack.  
Otherwise when PB-PJ attack has been successfully  identified in quantum learning layer, i.e.,  $B_j=1$ for $1 \le i \le N_{\rm E}$, 
Alice changes  the $j$-th  code digit of ${{\bf{b}}_{\rm{A}}}$ to be zero. Based on the decoding principle of nonrandom binary superimposed code,  Alice  can derive a new binary codeword  enabled to be decomposed as  ${\bf{b}}_i,  i \in \cal K $. Using this method, Alice can know the number and identities of  ALUs  and their codewords precisely under PB-PJ attack. We can say that UAD can be performed precisely  along with the codeword decoding.

\section{Design Principle of Quantum Learning Layer}
\label{DPQLL}
%
In this section, we  focus on the function ${g_3\left( x \right)}$ in quantum learning layer and  reveal its intrinsic potential for the support of  quantum learning to capture the uncertainty of codewords from Ava. 

\begin{figure}[!t]
	\centering \includegraphics[width=1.00\linewidth]{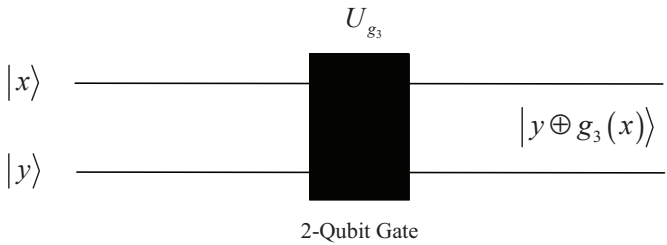}
	\caption{Illustration  of  quantum circuit using  quantum membership oracle. }
	\label{Quantum_phase}
	\vspace{-10pt}
\end{figure}
\begin{figure}[!t]
	\centering \includegraphics[width=1.00\linewidth]{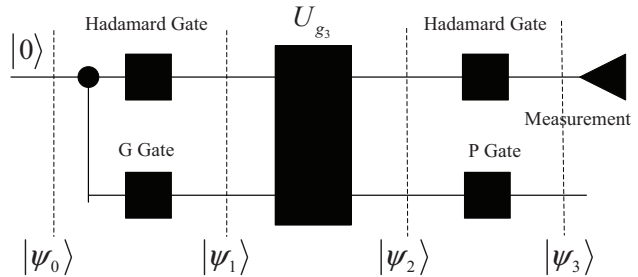}
	\caption{Illustration  of the quantum circuit to  implement a quantum learning  algorithm.}
	\label{Quantum_Box}
	\vspace{-10pt}
\end{figure}
\subsection{The Property of Function ${g_3\left( x \right)}$}
Based on the decoding principle discussed  in the previous layer, we can show the structure of ${g_3\left( x \right)}$ in Fig.~\ref{Black_Box}. Since ${{\bf{b}}_{\rm{A}}}$ and ${\bf{a}} \left( {{j}} \right)$ are random, $y_1$, $y_2$ and $y_3$ are basically  random binary digits. Note that the values of $x$ are also  binary digits. Therefore, we can have the following proposition.    
\begin{proposition}
	${g_3\left( x \right)}$ is  a black-box boolean function which aims to compute the following function:
	\begin{equation}
		g_3:\left\{ {0,1} \right\} \to \left\{ {0,1} \right\}
	\end{equation}
\end{proposition}
\begin{theorem}
	If ${g_3\left( x \right)}$ is a balanced function,  the $j$-th codeword digit belongs to ALU. Otherwise when ${g_3\left( x \right)}$ is a constant function, the $j$-th codeword digit belongs to Ava. 
\end{theorem}
\begin{proof}
	In principle, a boolean function can be either a constant function or a balanced function. 
	In order to determine whether the $j$-th digit belongs to Ava, Alice  has to calculate the function $g_3\left( 0 \right) \oplus g_3\left( 1 \right)$. If $g_3\left( 0 \right) \oplus g_3\left( 1 \right) = 0$, the  boolean function $g_3$ is  balanced and Alice   confirms that the $j$-th codeword digit belongs to Ava undoubtedly. Otherwise if $g_3\left( 0 \right) \oplus g_3\left( 1 \right) = 1$, the  boolean function $g_3$ is  constant and the $j$-th codeword digit  belongs to ALUs.
\end{proof}
The problem of identifying codeword digits from Ava has been transformed into an issue of how to compute the black-box boolean  function ${g_3\left( x \right)}$. 
\begin{figure*}[!t]
	\centering \includegraphics[width=1\linewidth]{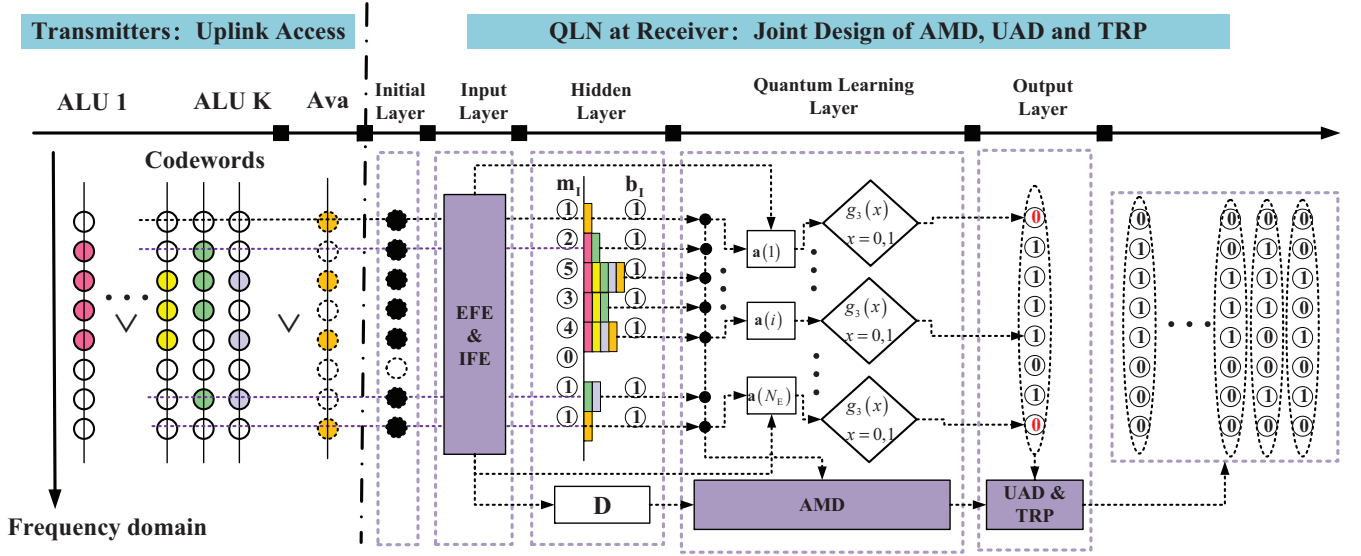}
	\caption{ QLN based decoding process of nonrandom superimposed code on SAPs. }	
	\label{QLN_detail}
	\vspace{-15pt}
\end{figure*}
\subsection{Quantum Learning Algorithm}
We in this subsection  explore quantum membership oracle to support  quantum learning  to learn whether the function  ${g_3\left( x \right)}$ is constant or balanced.  A quantum membership oracle ${U_{g_3}}$ is a unitary transformation that acts on the computational basis states as:
\begin{equation}
	{U_{g_3}}:\left| x \right\rangle \left| y \right\rangle   \mapsto {\rm{ }}{{U_{g_3}}}\left| x \right\rangle \left| y \right\rangle \mapsto \left| x \right\rangle \left| {y \oplus {g_3}\left( x \right)} \right\rangle .
\end{equation}
where ${U_{g_3}}$ denotes a quantum circuit for $g_3$ on the second qubit, shown in Fig.~\ref{Quantum_phase}. $\left| x \right\rangle$ denotes  the first input qubit  while $\left| y \right\rangle$  denotes   the second one.
When the second input qubit  is configured to be in the
state $\left| y \right\rangle  = \left| 0 \right\rangle $, then $\left| x \right\rangle  = \left| 0 \right\rangle $ in the first input qubit  will give $\left| {0 \oplus f\left( 0 \right)} \right\rangle  = \left| {f\left( 0 \right)} \right\rangle $ in the second output bit.  Similarly,  $\left| x \right\rangle  = \left| 1\right\rangle $ in the first input qubit will give $\left| {f\left( 1 \right)} \right\rangle$.  Therefore,  $\left| x \right\rangle  = \left| 0\right\rangle $ can be deemed  as a quantum version of the (classical) input bit 0, and $\left| x \right\rangle  = \left| 1\right\rangle $  as a quantum version of the input bit 1.
If the quantum circuit is expected to resolve a specific classical problem, the input and output should be designed deliberately. 

Now let us turn to the design of quantum circuit. We define the first qubit as control qubit and the second qubit  as  target register. For the control qubit,  two  Hadamard gates are configured, respectively at the input and output.  This can be seen in Fig.~\ref{Quantum_Box}.  Besides, we configure a G gate at the input of target register.  The input before the G gate is from the control qubit. The G gate is defined by ${\bf{G}}{\rm{ = }}\frac{1}{{\sqrt 2 }}\left[ {\begin{array}{*{20}{c}}
		1&1\\
		{ - 1}&1
\end{array}} \right]$. We also configure a P gate which is expressed as ${\bf{P}}{\rm{ = }}\frac{1}{{\sqrt 2 }}\left[ {\begin{array}{*{20}{c}}
		1&-1\\
		{ - 1}&-1
\end{array}} \right]$ at the second  output. 

In what follows,  we  specify  the working principle of quantum circuit  using   four 2-qubit basis vectors, i.e., $\left| {{\psi _0}} \right\rangle$, $\left| {{\psi _1}} \right\rangle$, $\left| {{\psi _2}} \right\rangle$ and $\left| {{\psi _3}} \right\rangle$. Note that $\left| {{\psi _0}} \right\rangle$, $\left| {{\psi _1}} \right\rangle$ and $\left| {{\psi _2}} \right\rangle$ are just the imitate representation  of quantum state  for the sake of  our understanding, rather than  the  results in measurement,  since each measurement will change the quantum state.  Those bases can be expressed by:
\begin{equation}
	\left| {{\psi _0}} \right\rangle  = \left| 0 \right\rangle 
\end{equation}
\begin{eqnarray}
	\left| {{\psi _1}} \right\rangle  = \left( {\frac{1}{{\sqrt 2 }}\left| 0 \right\rangle  + \frac{1}{{\sqrt 2 }}\left| 1 \right\rangle } \right)\left(  {\frac{{\left| 0 \right\rangle  - \left| 1 \right\rangle }}{{\sqrt 2 }}}\right)
\end{eqnarray}
\begin{eqnarray}
	\hspace{-10pt}
	\left| {{\psi _2}} \right\rangle  \hspace{-5pt}&=&\hspace{-5pt} \frac{{{{\left( { - 1} \right)}^{{g_3}\left( 0 \right)}}}}{{\sqrt 2 }}\left| 0 \right\rangle \left( {\frac{{\left| 0 \right\rangle  - \left| 1 \right\rangle }}{{\sqrt 2 }}} \right) + \frac{{{{\left( { - 1} \right)}^{{g_3}\left( 1 \right)}}}}{{\sqrt 2 }}\left| 1 \right\rangle \left( {\frac{{\left| 0 \right\rangle  - \left| 1 \right\rangle }}{{\sqrt 2 }}} \right)   \nonumber\\
	\hspace{-5pt}&=&\hspace{-5pt} {\left( { - 1} \right)^{{g_3}\left( 0 \right)}}\left( {\frac{{\left| 0 \right\rangle  + {{\left( { - 1} \right)}^{{g_3}\left( 0 \right) \oplus {g_3}\left( 1 \right)}}\left| 1 \right\rangle }}{{\sqrt 2 }}} \right)\left( {\frac{{\left| 0 \right\rangle  - \left| 1 \right\rangle }}{{\sqrt 2 }}}\right)
\end{eqnarray}
Examining Eq. (35), we know that  when  ${g_3}\left( 0 \right) \oplus {g_3}\left( 1 \right) = 0$, the result  $\left| {{\psi _2}} \right\rangle  = {\left( { - 1} \right)^{{g_3}\left( 0 \right)}}\left( {\frac{{\left| 0 \right\rangle  + \left| 1 \right\rangle }}{{\sqrt 2 }}} \right)\left( {\frac{{\left| 0 \right\rangle  - \left| 1 \right\rangle }}{{\sqrt 2 }}} \right) $ holds true.  The final Hadamard gate  on the control qubit and and P gate on the target register  transform the state to:
\begin{equation}
	\left| {{\psi _3}} \right\rangle  = {\left( { - 1} \right)^{{g_3}\left( 0 \right)}}\left| 0 \right\rangle \left| 0 \right\rangle
\end{equation}
As we can see, we need to solely focus on the  measurement of control qubit. Obviously, the squared norm of the basis state $\left| 0 \right\rangle $ in the control qubit is 1,  meaning that when   ${g_3}\left( 0 \right) \oplus {g_3}\left( 1 \right) = 0$  one measurement of the control qubit is certain to return the value 0. Similarly, when  ${g_3}\left( 0 \right) \oplus {g_3}\left( 1 \right) = 1$  one measurement of the control qubit is certain to return the value 1. 
One measurement of the control  qubit  can directly determine the value of ${g_3}\left( 0 \right) \oplus {g_3}\left( 1 \right) $ . Therefore, Alice using above quantum circuit  can  decide whether or not the codeword digit at $j$ belongs to Ava  at the cost of only half of the query complexity of original nonrandom superimposed coding.

In our design scheme, a quantum circuit with a single qubit is sufficient to implement an quantum oracle. Note that the scheme in~\cite{Z_Gedik} also proposed an example for quantum computation task with a single-qubit  input. However, it requires  at least  a  three-level quantum system and has special requirements on the computational task, i.e., determining parity of cyclic permutations. Different from that, our scheme is independent with the computational task even for a two-level quantum system, which can be well generalized to higher dimensional cases. 

Finally, we provide the overall implementation   details of QLN  in decoding process of nonrandom superimposed code on SAPs in Fig.~\ref{QLN_detail}.

\section{Performance Analysis}
\label{PA}
In the previous sections, we have detailed the principle of joint design  of AMD, UAD and TRP in QLN, and  also shown how to eliminate the effect of pilot-aware attack completely.  Benefiting from this, the new grant-free URLLC system under pilot-aware attack  can maintain normal channel estimation and data decoding  by consuming $K+2$ OFDM symbols and $N_{\rm E}$ pilot subcarriers. However, we still need to analyze  the reliability performance of this grant-free URLLC system.  To this end, we consider using the  channel estimation model shown in Eq. (9) and the data transmission model shown in Eq. (13). The goal is to  derive the expressions of failure probability for the grant-free URLLC system. 

Basically, a certain probability of failure occurs in uplink short packet data transmission  due to the  inevitable  decoding error of short packets.  There has been a well-known expression in literature~\cite{Berardinelli} to  calculate the decoding error probability $P_{d}$  of  transmissions over fading channels as a function of the average received SNR $\gamma_{0}$, the transmission rate $R$ and  matched filter  receiver, given by:
\begin{equation}
	{P_{\rm d}} = \int_0^\infty  {{Q\left( {\frac{{C\left( x  \right) - R}}{{\sqrt {{{V\left( x  \right)} \mathord{\left/
									{\vphantom {{V\left( x \right)} {\left( {{N_{\rm D}}m_{\rm D}} \right)}}} \right.
									\kern-\nulldelimiterspace} {\left( {{N_{\rm D}}m_{\rm D}} \right)}}} }}} \right)}} {f_{{K_c}}}\left( x \right)dx
\end{equation} 
where  $Q\left( x \right) = \int_x^\infty  {\frac{1}{{\sqrt {2\pi } }}} {e^{ - \frac{{{t^2}}}{2}}}dt$, $C\left( x \right) = {\log _2}\left( {1 + x } \right)$ and $V\left( x \right) = 1 - \frac{1}{{{{\left( {1 + x } \right)}^2}}}$.  On each subcarrier within TFRG$\#$3, $K_{c}$ interfering signals coexist  and there exists $K_{c}=K-1$. Without loss of generality,   we assume interfering  signals are Gaussian distributed  and  received at Alice with the same average SNR $\gamma _0$. The distribution ${f_{{K_c}}}\left( x \right) $ under matched filter receiver satisfies:
\begin{equation}
	{f_{{K_c}}}\left( x \right) = \frac{{{x^{{N_T} - 1}}{e^{ - \frac{x}{{{\gamma _0}}}}}}}{{\left( {{N_T} - 1} \right)!\gamma _0^{{K_c} + 1}}}\sum\limits_{i = 0}^{{N_T}} {\left( {\begin{array}{*{20}{c}}
				{{N_T}}\\
				i
		\end{array}} \right)} \frac{{\gamma _0^{{K_c} + i}\Gamma \left( {{K_c} + i} \right)}}{{\Gamma \left( {{K_c}} \right){{\left( {x + 1} \right)}^{{K_c} + i}}}}
\end{equation}
where ${\Gamma \left( \cdot \right)}$ denotes the Gamma function. 

In this paper, we assume at most one retransmission can be supported in the uplink.  Since TRP and UAD are both precise, the first transmission is deemed successful if its data is decoded successfully. In this case, the probability of correct data  decoding   is ${1 - {P_d}}$. Otherwise, the ALU would perform a retransmission over shared resources. The probability of correctly decoding the retransmitted data can be calculated by ${P_d}\left( {1 - {P_d}} \right)$. Finally the failure probability of grant-free URLLC system, denoted by $P_{\rm e}$, is given by:
\begin{equation}
	{P_{e}} = 1 - \left( {1 - {P_d}} \right) - {P_d}\left( {1 - {P_d}} \right) = P_d^2
\end{equation}
%

\begin{theorem}
	With precise CSI estimated, the asymptotic expression of received signal-to-interference-plus-
	noise ratio (SINR)  at Alice as $N_{\rm T}\to \infty $ is given by:
	\begin{equation}
		{\gamma _{asy}} \triangleq \gamma _{asy}^{{\text{perfect}}} \xrightarrow[N_{\rm T}\to \infty]{{\rm a.s.}} \frac{{{N_{\text{T}}}{\gamma _0}}}{{{\gamma _0}{K_c} + 1}}
	\end{equation}
	and the result with estimation error is given by:
	\begin{equation}
		{\gamma _{asy}} \triangleq \gamma _{asy}^{{\text{error}}} \xrightarrow[N_{\rm T}\to \infty]{{\rm a.s.}} \frac{{{N_{\text{T}}}{\gamma _0}\left( {1 - \lambda } \right)}}{{{\gamma _0}{K_c} + \lambda {\gamma _0} + 1}},0 < \lambda  < 1
	\end{equation}
	The decoding error probability satisfies:
	\begin{equation}
		{P_{\text{d}}} \xrightarrow[N_{\rm T}\to \infty]{{\rm a.s.}}  Q\left( {\frac{{C\left( {{\gamma _{asy}}} \right) - R}}{{\sqrt {{{V\left( {{\gamma _{asy}}} \right)} \mathord{\left/
								{\vphantom {{V\left( {{\gamma _{asy}}} \right)} {\left( {{N_{\text{D}}}{m_{\text{D}}}} \right)}}} \right.
								\kern-\nulldelimiterspace} {\left( {{N_{\text{D}}}{m_{\text{D}}}} \right)}}} }}} \right)
	\end{equation}
	and the failure probability of grant-free URLLC  can be finally expressed as Eq.~(43). The specific value of ${\gamma _{asy}}$ depends on the estimation assumption above. 
	\begin{figure*}[ht]
		\begin{equation}
			\label{Eq_23}
			{P_{\text{e}}} \xrightarrow[N_{\rm T}\to \infty]{{\rm a.s.}}   {\left\{ { { Q\left( {\frac{{\left( {1 + {\gamma _{asy}}} \right)\left[ {{{\log }_2}\left( {1 + {\gamma _{asy}}} \right) - \frac{R}{{{m_{\text{D}}}{T_{\text{s}}}{N_{\text{D}}}\Delta f}}} \right]\sqrt {{N_{\text{D}}}{m_{\text{D}}}{T_{\text{s}}}} }}{{\sqrt {\left[ {{{\left( {1 + {\gamma _{asy}}} \right)}^2} - 1} \right]{T_{\text{s}}}} }}} \right)} } \right\}^2}
			\vspace{-5pt}
		\end{equation}
	\end{figure*}
\end{theorem}

\begin{proof}
	Eq. (13) represents the model of received signal from the $m$-th ALU. According to the strong law of large numbers~\cite{Hoydis}, the following equations hold true
	\begin{equation}{}
		\frac{1}{{{N_{\text{T}}}}}\sum\limits_{i = 1}^{{N_{\text{T}}}} {{{\left| {{g_{j,{m},i}}} \right|}^2}}  \xrightarrow[N_{\rm T}\to \infty]{{\rm a.s.}}  {\mathbb E}{\left| {{g_{j,{m},i}}} \right|^2} = 1
	\end{equation}
	\begin{equation}{}
		\sum\limits_{p \in {\cal K},p \ne m}^{} {\sum\limits_{i = 1}^{{N_{\rm{T}}}} {\frac{{\hat g_{j,m,i}^*{g_{j,p,i}}}}{{{N_{\rm{T}}}}}} }  \xrightarrow[N_{\rm T}\to \infty]{{\rm a.s.}}  {\mathbb E}\left[ {\widehat g_{j,{m},i}^*{g_{j,p,i}}} \right]_{{p\in \cal K}, p \ne m} = 0
	\end{equation}
	\begin{equation}{}
		\frac{1}{{{N_{\text{T}}}}}\sum\limits_{i = 1}^{{N_{\text{T}}}} {\widehat g_{j,{m},i}^*{w_{j,i}}\left[ k \right]}  \xrightarrow[N_{\rm T}\to \infty]{{\rm a.s.}}   {\mathbb E}\left[ {\widehat g_{j,{m},i}^*{w_{j,i}}\left[ k \right]} \right] = 0
	\end{equation}
	When ${\widehat {\mathbf{g}}_{j,{m}}} = {{\mathbf{g}}_{j,{m}}}$, the received SINR for the $i$-th ALU can be written as:
	\begin{equation}
		\gamma _{asy}^{{\text{perfect}}} \xrightarrow[N_{\rm T}\to \infty]{{\rm a.s.}}  \frac{{\frac{1}{{{N_{\text{T}}}}}\sum\limits_{i = 1}^{{N_{\text{T}}}} {{{\left| {{g_{j,{m},i}}} \right|}^2}} }}{{{{\frac{{\left( {\sum\limits_{{p \in {\cal K},p \ne m}} {\sum\limits_{i = 1}^{{N_{\text{T}}}} {\widehat g_{j,{m},i}^*{g_{j,p,i}}} } } \right)}}{{{N_{\text{T}}}\sum\limits_{i = 1}^{{N_{\text{T}}}} {{{\left| {{g_{j,{m},i}}} \right|}^2}} }}}^2} + \frac{1}{{{N_{\text{T}}}{\gamma _0}}}}}
	\end{equation}
	Eq. (40) can be derived after simplification.  The received SINR for the $i$-th ALU under estimation errors can be written as:
	\begin{equation}
		\gamma _{asy}^{{\text{error}}}\xrightarrow[N_{\rm T}\to \infty]{{\rm a.s.}} \frac{{\frac{1}{{{N_{\text{T}}}}}\sum\limits_{i = 1}^{{N_{\text{T}}}} {\widehat g_{j,{m},i}^*{g_{j,{m},i}}} }}{{{{\frac{{\left( {\sum\limits_{{p \in {\cal K},p \ne m}} {\sum\limits_{i = 1}^{{N_{\text{T}}}} {\widehat g_{j,{m},i}^*{g_{j,p,i}}} } } \right)}}{{{N_{\text{T}}}\sum\limits_{i = 1}^{{N_{\text{T}}}} {\widehat g_{j,{m},i}^*{g_{j,{m},i}}} }}}^2} + \frac{1}{{{N_{\text{T}}}{\gamma _0}}}}}
	\end{equation}
	where ${\widehat g_{j,{m},i}} = \left( {1 - \lambda } \right){g_{j,{m},i}} - \lambda {\widetilde g_{j,{m},i}},0 < \lambda  < 1$.  After  simplification, we can derive the Eq. (41). Due to the disappearance of channel randomness under large number of  antennas, we can  calculate Eq.~(37)  as Eq.~(42). By substituting Eq.~(42) into Eq.~(39), we can derive  Eq.~(43). The proof is complete. 	
\end{proof}

 \section{Numerical Results}
 \label{NR}

\begin{table}[!t]\footnotesize
	\caption{\label{tab_testurllc}Simulation Parameters and Values}
	\vspace{-15pt}
	\begin{center}
		\footnotesize
		\begin{tabular*}{9cm}{@{\extracolsep{\fill}}ll}
			\toprule
			\multicolumn{1}{c}{ Simulation Parameters}  & \multicolumn{1}{c}{ Values }\\
			\midrule
			\multicolumn{1}{c}{Modulation}&\multicolumn{1}{c}{OFDM }
			\\
			\multicolumn{1}{c}{Subcarrier  spacing  }&\multicolumn{1}{c}{$\Delta f=$60kHz,120kHz}
			\\
			\multicolumn{1}{c}{Channel bandwidth }&\multicolumn{1}{c}{$\le$ 100MHz, $\le$400MHz}
			\\
			\multicolumn{1}{c}{Coherence bandwidth}&\multicolumn{1}{c}{$3 \times \Delta f$}
			\\
			\multicolumn{1}{c}{OFDM symbol duration }&\multicolumn{1}{c}{${T_{\rm{s}}}=17.84 {\rm {\mu s}}, 8.93 {\rm {\mu s}}$}
			\\
			\multicolumn{1}{c}{Number of OFDM symbols for TRP}&\multicolumn{1}{c}{$m_{\rm E} = K+2$}
			\\
			\multicolumn{1}{c}{Number of  subcarriers for TRP }&\multicolumn{1}{c}{$N_{\rm E}\le 1024$ }
			\\
			\multicolumn{1}{c}{Number of  subcarriers for channel estimation}&\multicolumn{1}{c}{${ N}_{\rm CE}=128$ }
			\\
			\multicolumn{1}{c}{Pilot subcarrier arrangement for TRP}&\multicolumn{1}{c}{ Proposed Coding}
			\\
			\multicolumn{1}{c}{Pilot subcarrier arrangement for  channel estimation}&\multicolumn{1}{c}{Block type }
			\\
			\multicolumn{1}{c}{Channel  estimator at gNB}&\multicolumn{1}{c}{Least Square}
			\\
			\multicolumn{1}{c}{ Number of subcarriers for data transmission}&\multicolumn{1}{c}{$N_{\rm D}= 4$ }
			\\
			\multicolumn{1}{c}{Combing technique at gNB }&\multicolumn{1}{c}{ Matched filter receiver}
			\\
			\multicolumn{1}{c}{ Size of data packets}&\multicolumn{1}{c}{$R=32~{\rm Bytes}$ }
			\\
			\multicolumn{1}{c}{ SINR of receiving  data at gNB}&\multicolumn{1}{c}{$\gamma =10 {\rm dB}$ }
			\\
			\multicolumn{1}{c}{Time consumed by other operations }&\multicolumn{1}{c}{${T_{\rm{extra}}}=100, 300~{\rm {\mu s}}$ }
			\\
			\multicolumn{1}{c}{Number of antennas at gNB}&\multicolumn{1}{c}{ $N_{\rm T}= 128$}
			\\
			\multicolumn{1}{c}{ Channel fading model}&\multicolumn{1}{c}{Rayleigh}
			\\
			\multicolumn{1}{c}{Number of ALUs}&\multicolumn{1}{c}{$K\le 32$}
			\\
			\multicolumn{1}{c}{Number of channel taps}&\multicolumn{1}{c}{$L=6$}
			\\
			\multicolumn{1}{c}{Failure probability requirement }&\multicolumn{1}{c}{$P_{\rm e}\le 10^{-5}$}
			\\
			\multicolumn{1}{c}{Latency constraints}&\multicolumn{1}{c}{$T_{\rm con}=1{\rm ms}$}
			\\
			\bottomrule
		\end{tabular*}
	\end{center}
	\vspace{-15pt}
\end{table}
In this section, we will evaluate the  performance of uplink access of grant-free URLLC using quantum  learning  based nonrandom  superimposed coding, This refers to several metrics, including the code rate, overheads, reliability and latency under channel estimation errors. Frequency range 1 (FR1) for Sub-6 GHz and Frequency range 2 (FR2) for millimeter wave in 5G NR are considered respectively. The system is expected to work within at most 100 MHz channel bandwidth for  FR1 and at most 400 MHz channel bandwidth for FR2~\cite{TS38_300}.  Simulation parameters and values can be seen in Table~\ref{tab_testurllc}.  Note that one pilot every three consecutive subcarriers is inserted  to  acquire  independent frequency-domain variations of the channels. Given above bandwidth constraints,  the maximum number of  pilot subcarriers for TRP is limited to 512. Therefore,  at most $512 \times 3 \times 240=368.84\rm MHz $  channel bandwidth is occupied when $\Delta f=120 \rm kHz$ and at most $512 \times 3 \times 64=98.304 \rm MHz $  channel bandwidth is required when $\Delta f=60 \rm kHz$. 

\begin{figure}[!t]
	\centering \includegraphics[width=1\linewidth]{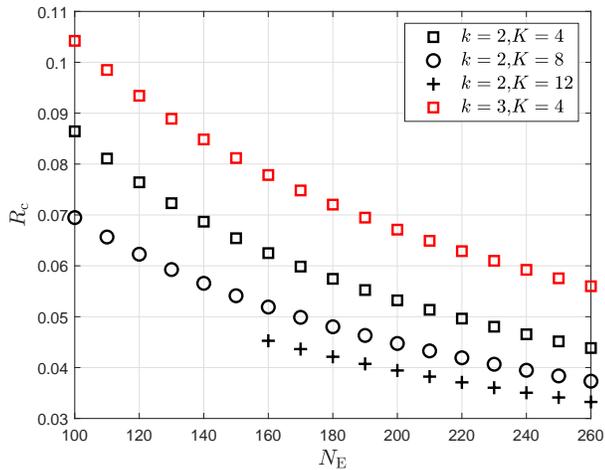}
	\caption{ The code rate $ R_{\rm c}$ versus $N_{\rm E}$ under various $k$ and $K$; . }	
	\label{Coderate}
	\vspace{-15pt}
\end{figure}
Fig.~\ref{Coderate} presents the curve of code rate $ R_{\rm c}$  versus $N_{\rm E}$. As we can see, increasing $K$ and $N_{\rm E}$ will reduce  the code rate;  for example, $ R_{\rm c}$  decreases from 0.06 to 0.05 if  $K$ increases from 4 to 8 at  $k=2$ and $N_{\rm E}=170$. Furthermore, increasing  $k$ will  increase the code rate; for example, $ R_{\rm c}$  increases from 0.07 to 0.085 if $k$ increases from 2 to 3 when  $K=4$ and $N_{\rm E}=140$; Increasing  $K$  will also increase  $N_{\rm E}$ since the lower bound of  available $N_{\rm E}$, that is, $K\left( {k - 1} \right)\left[ {1 + K\left( {k - 1} \right)} \right]$,   increases with the increase of $K$; for example, when $k=2$, the increase of $K$ from 8 to 12 will increase the lower bound of $N_{\rm E}$   from 72 to 156. We can also find that the influence of proposed  scheme on frequency-domain resource overheads is  greater than that on code-domain resource overheads, which means that higher degree of optimization on code rate can be performed to reduce the subcarrier resource  consumption, given a certain level of reliability. 

\begin{figure}[!t]
	\centering \includegraphics[width=1\linewidth]{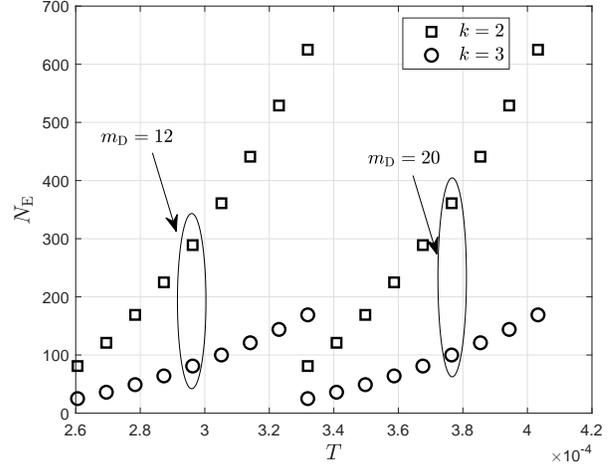}
	\caption{ Time-frequency domain resource overheads. }	
	\label{Subcarrier_latency}
	\vspace{-15pt}
\end{figure}
In order to clearly investigate the influence of the proposed scheme  on time-frequency resource overheads, we in  Fig.~\ref{Subcarrier_latency} present the curves of $N_{\rm E}$ versus $T$, respectively  under  $m_{\rm D}=12$ and $m_{\rm D}=20$.  The parameter $T$  satisfies $T = \left( {m_{\rm E}+ m_{\rm D}} \right)\times {T_{\rm s}} + {T_{\rm  extra}}$. ${T_{\rm  extra}}$ is fixed to be $100 ~{\rm {\mu s}}$ and ${T_{\rm{s}}}=8.93 {\rm {\mu s}}$ is configured under $\Delta f=120 \rm kHz$. When $K$ increases from 4 to 12, the total latency $T$ and the subcarrier overheads $N_{\rm E}$ both increase.  This tendency does not change under various  $m_{\rm D}$, meaning that both  time-domain and frequency-domain resources are critically important for pilot protection. Indeed,  pilot subcarriers under $k=2$ are consumed more with respect  to  those under  $k=3$.  This is because larger $k$ can bring  higher coding diversity on code domain and thus reduce the  consumption of frequency-domain resources. Fig.~\ref{Subcarrier_latency}  also verifies the fact that  the influence of proposed scheme on frequency-domain resource overheads is  greater than that on time-domain resource overheads.

\begin{figure}[!t]
	\centering \includegraphics[width=1\linewidth]{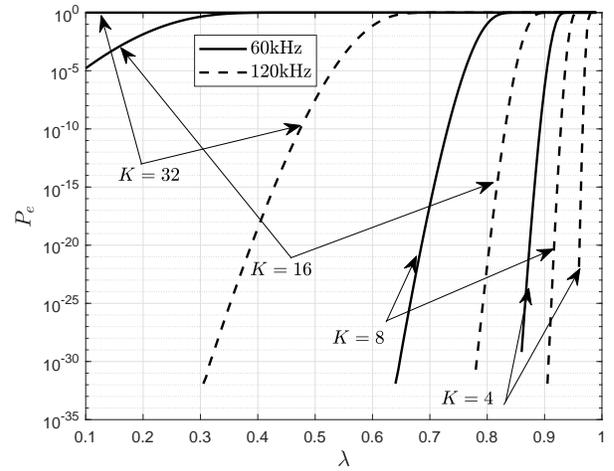}
	\caption{ Failure probability versus channel estimation errors. }	
	\label{Reliability_estimation}
	\vspace{-15pt}
\end{figure}
Fig.~\ref{Reliability_estimation} shows the curves of failure probability versus  channel estimation error $\lambda$.  $\Delta f=60 \rm kHz$ with  $T_{s}=17.84 \times 10^{-6} \rm s$ and  $\Delta f=120 \rm kHz$ with  $T_{s}= \times 10^{-6} \rm s$   are respectively configured under $N_{\rm E}=128$ and $\gamma_{0}=0.1$.  $m_{\text{D}}$  is determined by ${m_{\text{D}}} = {{\left( {{T_{{\text{con}}}} - {m_{\text{E}}}{T_s} - {T_{{\text{extra}}}}} \right)} \mathord{\left/
		{\vphantom {{\left( {{T_{{\text{con}}}} - {m_{\text{E}}}{T_s} - {T_{{\text{extra}}}}} \right)} {{T_s}}}} \right.
		\kern-\nulldelimiterspace} {{T_s}}}$ with ${T_{\rm{extra}}}=300{\rm {\mu s}}$. 
As we can see, $P_{\rm e}$  increases  with the  increase of  $\lambda$ if $\lambda$ is above a certain threshold which changes with the number  $K$ of ALUs. 
When  $K$  increases, the fluctuation of this threshold is more sensitive to the  changes  of  channel estimation errors.  Since channel estimation errors  cannot be eliminated completely,  more ALUs  would  introduce more intra-user interference and  disturbance caused by imprecise  channel estimation.  Especially when $K$ is configured to be larger than 16, FR1 cannot support proposed new grant-free URLLC system. 
\begin{figure}[!t]
	\centering \includegraphics[width=1\linewidth]{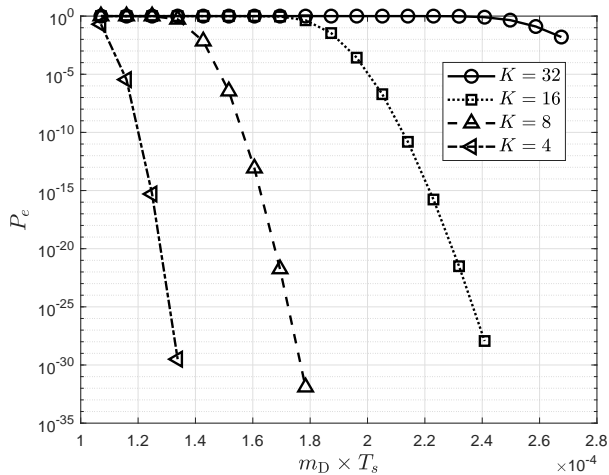}
	\caption{ Failure probability versus the  transmission latency. }	
	\label{Reliability_latency_Novel}
	\vspace{-15pt}
\end{figure}

Fig.~\ref{Reliability_latency_Novel} depicts  the  reliability-latency tradeoff curves under various $K$.  $N_{\rm E}$ is configured to be 128 and  $m_{\text{D}}$ is selected from $12 \times T_{\rm s}$ to $30 \times T_{\rm s}$. ${T_{\rm  extra}}$ is fixed to be $300 ~{\rm {\mu s}}$ and ${T_{\rm{s}}}$ is configured to be $8.93 {\rm {\mu s}}$ under $\Delta f=120 \rm kHz$.   $\lambda$ is chosen as $0.2$.  We can see that data transmission latency   should not be lower than a certain threshold if it is expected to achieve the reliability of $99.999\% $  or more. With the increase of $K$, the failure probability is increased and the reliability is decreased. Especially when $K$ is equal to 32, the system reliability  can hardly satisfy the requirements of 5G URLLC. Those hints provides guiding principles for designing and implementing grant-free URLLC system under pilot-aware attack.

\section{Conclusions}
\label{Conclusions}
In this paper,  we proposed a quantum  learning  based nonrandom  superimposed coding  method  to  support flexible, secure and efficient SAP encoding/decoding such that those SAPs can  be used for  TRP in uplink access process of grant-free URLLC systems threated by a pilot-aware attack. This design changed the main access procedures as those including AMD, UAD,  TRP,  channel estimation and data transmission. A quantum learning network was designed to learn the uncertainty of attack on SAP decoding precisely. With the help of this network, AMD, UAD and TRP could be precisely realized and meanwhile  the performance of channel estimation and data transmission can be recovered to a normal level.  By considering those procedures together, we derived novel expressions of failure probability of the new  grant-free URLLC system to evaluate its reliability. With above efforts, we can say that it is feasible to use quantum  learning to help information coding  defend against unknown attack on itself, which is an interesting  direction for future research.

	\begin{IEEEbiography}[{\includegraphics[width=1in,height=1.25in,clip,keepaspectratio]{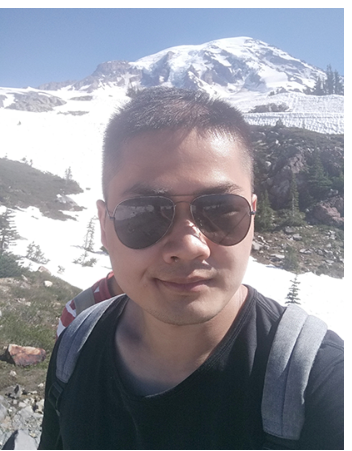}}]
	{Dongyang Xu} (M'19)  received the Ph.D. degree in Information and Communications Engineering in 2019 from Xi'an Jiaotong University, China.  Since 2019, he has been appointed as assistant professor of the School of Information and Communications Engineering at Xi'an Jiaotong University, Shaanxi 710049, China.  From January 2017 to January 2018, He was a visiting scholar under the supervision of Prof. James. A. Ritcey  at Department of Electrical $\&$ Computer Engineering, University of Washington, Seattle, USA. His current research interests  include   information security, coding theory, 5G/6G and quantum information science. He has published 30 technical papers on international journals and conferences. He received the Best Paper Rewards from  IEEE \textsc{China Communications} in 2017. He also served as the Technical Program Committee Member for IEEE/CIC ICCC in 2017 and WCSP 2019.  He is a  Member of IEEE and IEEE Communications Society.
\end{IEEEbiography}
\begin{IEEEbiography}
	[{\includegraphics[width=1in,height=1.25in,clip,keepaspectratio]{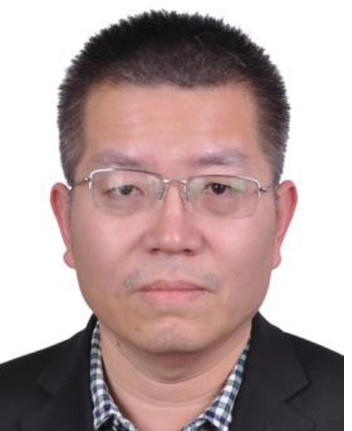}}]{Pinyi Ren} (M'10) received the Ph.D. degree in Electronic and Communications System, the M.S. degree in Information and Communications Engineering, the B.S. degree in Information and Control Engineering, in 2001, 1997, and 1994, respectively, all from Xi'an Jiaotong University, China. He is currently a Professor of the Information and Communications Engineering Department, Xi'an Jiaotong University, China. He has published over 100 technical papers on international Journals and conferences. He received the Best Letter Award of IEICE Communications Society 2010. He has over 30 Patents (First Inventor) authorized by Chinese Government. Dr. Pinyi Ren serves as an Editor for the Journal of Xi¡¯an Jiaotong University, an Editor for the Journal of Electronics and Information Technology, and has served as the Leading Guest Editor for the Special Issue of Mobile Networks and Applications on "Distributed Wireless Networks and Services" and the Leading Guest Editor for the Special Issues of Journal of Electronics on "Cognitive Radio". He has served as the General Chair of ICST WICON 2011, and frequently serves as the Technical Program Committee members of IEEE GLOBECOM, IEEE ICC, IEEE CCNC, etc. Dr. Pinyi Ren is a Member of IEEE and IEEE Communications Society.
\end{IEEEbiography}
\end{document}